\documentclass{article}
\usepackage{amssymb}
\usepackage{amsmath}
\usepackage{amsthm}
\usepackage{mathtools}
\usepackage{todonotes}
\usepackage{hyperref}
\usepackage{enumerate}
\usepackage{caption}
\usepackage{subcaption,graphicx}
\usepackage{tikz}
\usetikzlibrary{quantikz}
\usepackage{tikz-cd}
\usepackage{csquotes}

\usepackage{geometry}
\usepackage{mathptmx} 
\usepackage{authblk}
\usepackage{multicol}
\usepackage{pgfplots}

\allowdisplaybreaks

\usepackage{float}

\usepackage[ruled]{algorithm2e}
\usepackage[page]{appendix}

\newenvironment{FigureInColumn}
  {\par\medskip\noindent\minipage{\linewidth}\captionsetup{width=0.9\linewidth}}
  {\endminipage\par\medskip}

\newtheorem{theorem}{Theorem}
\newtheorem{lemma}[theorem]{Lemma}

\theoremstyle{definition}

\newcommand{\algoCap}{Kernel Descent}
\newcommand{\algo}{kernel descent}
\newcommand{\algoAbb}{KD}

\makeatletter
\renewcommand{\todo}[2][]{\tikzexternaldisable\@todo[#1]{#2}\tikzexternalenable}
\makeatother

\tikzset{
    operator/.append style={fill=blue!10},
}

\newlength{\Ubraceignoreleft}
\newlength{\Ubraceignoreright}

\NewDocumentCommand{\Ubrace}{ommo}{%
  {
   \IfValueT{#1}{\settowidth{\Ubraceignoreleft}{$#1$}}%
   \IfValueT{#4}{\settowidth{\Ubraceignoreright}{$#4$}}%
   \hspace*{\Ubraceignoreleft}
   \underbrace{%
     \hspace*{-\Ubraceignoreleft}%
     #2%
     \hspace*{-\Ubraceignoreright}%
   }_{#3}%
   \hspace*{\Ubraceignoreright}%
  }%
}

\usepackage{xcolor}

\title{Introducing the Kernel Descent Optimizer for Variational Quantum Algorithms}

 \author{Lars Simon\footnote{\tt lars.simon@bdr.de}}
 \author{Holger Eble\footnote{\tt holger.eble@bdr.de}}
 \author{Manuel Radons\footnote{\tt manuel.radons@bdr.de}}

\affil{\emph{Bundesdruckerei GmbH, Kommandantenstraße 18, 10969 Berlin, Germany}}

\date{}

\begin{document}
	
\maketitle

    \vspace{-1.0cm}




\begin{abstract}
In recent years, variational quantum algorithms have garnered significant attention as a candidate approach for near-term quantum advantage using noisy intermediate-scale quantum (NISQ) devices. In this article we introduce kernel descent, a novel algorithm for minimizing the functions underlying variational quantum algorithms. We compare kernel descent to existing methods and carry out extensive experiments to demonstrate its effectiveness. In particular, we showcase scenarios in which kernel descent outperforms gradient descent and quantum analytic descent. The algorithm follows the well-established scheme of iteratively computing classical local approximations to the objective function and subsequently executing several classical optimization steps with respect to the former. Kernel descent sets itself apart with its employment of reproducing kernel Hilbert space techniques in the construction of the local approximations, which leads to the observed advantages. 
	\end{abstract}

\begin{multicols}{2}

    \section{INTRODUCTION}
    \label{sec:introduction}
    
   The promise of exponential computational advantages from quantum computing remains unrealized in practice due to the noisiness and limited qubit count of currently available quantum hardware. Variational quantum algorithms (VQAs) have risen to prominence as a promising strategy to achieve utility of quantum computers in the near term, within the constraints of the so-called noisy intermediate-scale quantum (NISQ) era. 

VQAs leverage a classical optimizer to train a parametrized quantum circuit, drawing parallels to classical machine learning models such as neural networks \cite{Cerezo2021}. The primary objective of VQA training, in its most basic form, is to find a choice of parameters for a parametrized quantum circuit that minimizes the expected value of a given observable with respect to the state computed by the circuit. This adaptive approach, combined with the typically shallower circuit depths required compared to algorithms designed for fault-tolerant quantum computing, accounts for the expected NISQability of the approach \cite{Cerezo2021}. 

A common strategy for the minimization of the objective function underlying a VQA involves computing a local approximation around a given point in parameter space,
followed by one or more classical optimization steps with respect to this approximation. The outcome of this procedure yields the next base point for a local approximation of the function and the process is repeated until a predefined convergence criterion is met. 
Arguably, the most prominent representatives of this training philosophy are gradient descent (via linear approximations), and quantum analytic descent (via trigonometric approximations).  There also exist derivative-free methods ---for example, Nelder-Mead \cite{nelder_mead}, Powell \cite{powell_algorithm}, COBYLA \cite{COBYLA_algorithm}--- as well as methods designed specifically for the functions underlying VQAs, such as Rotosolve \cite{rotosolve_1}, \cite{rotosolve_2}, \cite{rotosolve_3}, \cite{rotosolve_4}. 
However, gradient-based optimizers are more common in the context of VQAs. The evaluations of the gradients are usually performed via the parameter-shift rules, even though the required number of circuit evaluations scales unfavorably with the number of trainable parameters, cf.  \cite{VQA_backpropagation_scaling}, \cite{abbas_backpropagation_state_tomography}.
For information on parameter-shift rules, see
\cite{Mitarai_2018}, \cite{Schuld_2019}, \cite{Mari_2021}, \cite{Wierichs2022generalparameter}. 

In this article, we introduce \algo, a novel training algorithm for VQAs. Kernel descent follows the iterative optimization process described above. Its key innovation is the exploitation of the fact that the functions computed by VQAs are contained in a set which naturally carries the structure of a reproducing kernel Hilbert space (RKHS).  
This approach often leads to an improved quality of the local approximations computed during training, to which we ascribe the algorithm's observed performance advantages, particularly in terms of speed of convergence and robustness.
We note that neither measurement shot noise nor quantum hardware noise are taken into account in the theoretical description of kernel descent. However, in \cite{simon2024denoising} we demonstrated experimentally that our RKHS techniques can in fact be used to mitigate the adverse effects of both these types of noise.

Kernel descent introduces a hyperparameter $L$ that governs the order of the local approximations computed during its execution. As $L$ increases, the approximation quality improves, albeit with a corresponding increase in the number of circuit evaluations per iteration. For $L=1$, the number of circuit evaluations per iteration is comparable to that of gradient descent using the parameter-shift rules, and for $L=2$, it is comparable to that of quantum analytic descent.
Consequently, our experiments will compare the performance of kernel descent with $L=1$ and $L=2$ to that of gradient descent and quantum analytic descent, respectively.

To level the playing field, all algorithms will be tested in their most basic form, without adaptive learning rates or similar enhancements. Advanced techniques, such as normalized gradient descent \cite{Hazan2015BeyondCS}, Nesterov’s Accelerated Gradient method \cite{Nesterov1983AMF}, the ADAM optimizer \cite{KingBa15}, or natural gradient descent \cite{natural_gradient_descent}, can be adapted to kernel descent (see \cite{Suzuki_2021} and \cite{quantum_natural_gradient} for their treatment in the context of VQAs). However, we leave this task for future investigations as it would go beyond the scope of this work.

The present article is part of a larger study on the use of RKHS techniques in the context of VQAs. 
This approach has already been employed in constructing classical surrogates of quantum machine learning models \cite{simon2024interpolating} and, as mentioned above, in alleviating the adverse effects of noise on gradient descent in VQAs \cite{simon2024denoising}.
The project originated in the application of RKHS methods to the training of neural support vector machines \cite{simon2023algorithms}, which subsequently led to the development and analysis of neural quantum support vector machines in \cite{simon2023neural}. This series of investigations laid the groundwork for the development of kernel descent as a robust optimization strategy specifically tailored to the functions underlying VQAs. 

\textbf{Content and Structure:} In Section \ref{sec:algorithm}, after laying out our assumptions and the optimization task, we introduce kernel descent and provide a detailed comparison of the algorithm's underlying theory with that of gradient descent and quantum analytic descent. Building on this groundwork, Section \ref{sec:experiments} presents the results of our experimental comparison of kernel descent to gradient descent and quantum analytic descent, providing empirical evidence for the claimed advantages in terms of speed of convergence and, particularly, robustness of the novel algorithm. Our closing remarks, including a discussion of the results and suggestions for future research directions, follow in Section \ref{sec:conclusion}.

	\section{ALGORITHM}
	\label{sec:algorithm}
    In Section \ref{subsec:alg_setting} we establish the setting and state the problem addressed by our algorithm; subsequently, in Section \ref{subsec:alg_prelim}, we go over some technical preliminaries. In Section \ref{subsec:alg_strategy} we state the general strategy used to tackle the problem, before giving a detailed description of our algorithm, including performance guarantees, in Section \ref{subsec:alg_description}. Finally, in Section \ref{subsec:alg_comparison}, we compare our algorithm to gradient descent and quantum analytic descent.
    
    \subsection{Setting}
    \label{subsec:alg_setting}
    Letting $m,n\in\mathbb{Z}_{>0}$, we consider functions of the form $f\colon\mathbb{R}^m\to\mathbb{R}$,
    \begin{equation}\label{eq:f}
        f(\theta) = \langle\psi (\theta)|\mathcal{M}|\psi (\theta)\rangle ,
    \end{equation}
    where $\mathcal{M}\in\mathbb{C}^{2^n\times 2^n}$ is an observable and
    $$|\psi (\theta)\rangle = C_{m+1}R_m (\theta_m) C_m\cdots  R_1 (\theta_1) C_1|0\rangle^{\otimes n},$$
    where $C_1 ,\dots , C_{m+1}$ are $n$-qubit unitaries and, for all $j\in\{1,\dots ,m\}$,
    $$R_j (\theta_j) = \exp \left(-i\frac{\theta_j}{2}G_j\right), \theta_j\in\mathbb{R},$$
    is a rotation, where $G_j\in\mathbb{C}^{2^n\times 2^n}$ is Hermitian with set of eigenvalues $\{-1,1\}$ (e.g., $G_j$ could be an element of $\{I,X,Y,Z\}^{\otimes n}\setminus\{I^{\otimes n}\}$, where $I,X,Y,Z$ denote the Pauli matrices). For a visual representation of the corresponding circuits, see Figure \ref{fig:general_circuit}.

    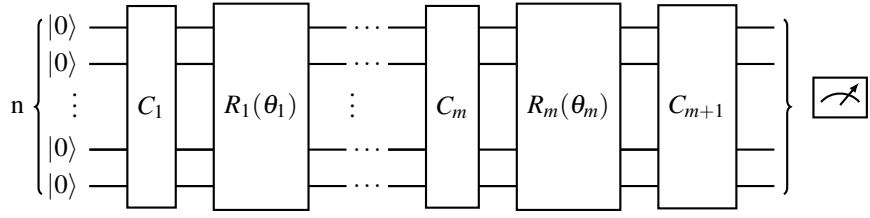
\begin{figure*}[htbp]
    \centering
        \begin{quantikz}
    \lstick[wires=5]{n}&\lstick{\ket{0}} & \gate[5, nwires=3]{C_1}    & \gate[5, nwires=3]{R_1(\theta_1)}   &\ \ldots\ \qw   & \gate[5, nwires=3]{C_m} &  \gate[5, nwires=3]{R_m(\theta_m)} & \gate[5, nwires=3]{C_{m+1}} & \qw \rstick[wires=5]{} \\[-10pt]
    & \lstick{\ket{0}} & \qw                 & \qw & \qw \ \ldots\ & \qw &  \qw & \qw &  \qw \\[-10pt]
     & \lstick{$\vdots$}  &       &           & \lstick{$\vdots$}     & &  &  &  &  |[meter]| \\[-3pt]
    & \lstick{\ket{0}} & \qw                & \qw    &\qw\ \ldots\ & \qw &   &\qw    &  \qw \\[-10pt]
    & \lstick{\ket{0}} & \qw                 & \qw & \qw\ \ldots\ & \qw &  \qw    &  &  \qw       
    \end{quantikz}
    \caption{This figure shows the circuit underlying the computation of the function $f$ in Section \ref{subsec:alg_setting}.}\label{fig:general_circuit}
    \end{figure*}

    The goal is now to minimize the function $f$, i.e., to find a point in parameter space $\mathbb{R}^m$, at which the value of $f$ is minimal. The most commonly-used algorithm for this task is gradient descent using the parameter-shift rules \cite{Mitarai_2018}, \cite{Schuld_2019}, \cite{Mari_2021}, \cite{Wierichs2022generalparameter}, which requires $2m$ circuit evaluations per iteration. Another algorithm is quantum analytic descent \cite{quantum_analytic_descent}, which requires $2m^2+m+1$ circuit evaluations per iteration. Both of these, as well as the algorithm we will introduce in Section \ref{subsec:alg_description}, are a special case of the technique introduced in Section \ref{subsec:alg_strategy}.

    \subsection{Technical preliminaries}
    \label{subsec:alg_prelim}
    It follows from the results in \cite{Schuld_2021} that $f$ (see Section \ref{subsec:alg_setting}) is contained in $H$, where $H$ is the set of all functions $g\colon\mathbb{R}^m\to\mathbb{R}$ of the form $g(z)=\sum_{\omega\in\{-1,0,1\}^m} c_\omega e^{i\omega^\intercal z}$ with $\overline{c_{-\omega}} = c_\omega\in \mathbb{C}$ for all $\omega\in\{-1,0,1\}^m$. When equipping $H$ with the (real) inner product given by
    \begin{align*}
    		\langle g_1 ,g_2\rangle_H= \int_{[-\pi ,\pi]^m} g_1 (z) g_2 (z) \mathrm{d}z, \ g_1,g_2\in H,
    \end{align*}
    it turns out that $H$ carries the structure of a reproducing kernel Hilbert space with kernel $K=\left(\frac{3}{2\pi}\right)^m\tilde{K}$, where
    \begin{align*}
    	\tilde{K}(x,z) = \prod_{j=1}^m \frac{1+2\cos(x_j - z_j )}{3}, \ x,z\in\mathbb{R}^m.
    \end{align*}
    The following lemma shows how the structure of $H$ can be exploited to obtain approximations of $f$:
    \begin{lemma}\label{lemma:general_approx}
        Let $D\in\mathbb{Z}_{\geq 1}$ and let $p_1,\dots ,p_D\in\mathbb{R}^m$. Then the linear system of equations
        \begin{align*}
				\left(\tilde K(p_i , p_j)\right)_{1\leq i,j\leq D}
				\cdot\tilde\eta
				=(f(p_1) ,\dots , f(p_D))^\intercal
    \end{align*}
        has at least one solution $\tilde\eta\in\mathbb{R}^D$, allowing us to define $\tilde{f}\colon\mathbb{R}^m\to\mathbb{R}$ given by
        $$\tilde{f}(\theta)= \sum_{j=1}^{D} \tilde\eta_j \tilde K(p_j, \theta).$$
        We then have:
        \begin{itemize}
            \item $\tilde{f}$ does not depend on the choice of (the not necessarily uniquely determined solution) $\tilde\eta$.
            \item $\tilde{f}$ is the orthogonal projection (wrt.\ $\langle\cdot , \cdot\rangle_H$) of $f$ onto the subspace of $H$ spanned by $K(p_1 , \cdot ), \dots , K(p_D , \cdot )$.
            \item $\tilde{f}(p_j) = f(p_j)$ for $j=1,\dots , D$.
        \end{itemize}
    \end{lemma}
    \begin{proof}
        This is obvious from the results in \cite{simon2024interpolating}, specifically Implementation Remark 6 and Lemma 13.
    \end{proof}

    In order to compute the approximation given by Lemma \ref{lemma:general_approx}, we need to solve a linear system of equations. However, for specific choices of the points $p_1,\dots ,p_D$, the occurring matrix becomes trivial to invert. One such choice is given by the following lemma:

    \begin{lemma}\label{lemma:orthonormal_basis}
        For all $p\in\mathbb{R}^m$, the family
        $$\left(\sqrt{\left(\frac{2\pi}{3}\right)^m}K(p+q,\cdot )\right)_{q\in\left\{-\frac{2\pi}{3},0,\frac{2\pi}{3}\right\}^m}$$
        is an orthonormal basis for $H$.
    \end{lemma}
    \begin{proof}
        Since $\dim_\mathbb{R}(H)=3^m = \operatorname{card}\left(\left\{-\frac{2\pi}{3},0,\frac{2\pi}{3}\right\}^m\right)$, the claim follows from explicitly calculating the pairwise inner products of the members of the given family.
    \end{proof}

    \subsection{Strategy}
    \label{subsec:alg_strategy}
    When aiming to minimize $f$, a common technique is to compute a local approximation around a given point in parameter space $\mathbb{R}^m$ and subsequently perform one (or several) optimization step(s) with respect to this approximation. More precisely, the technique consists of the following steps:
    \begin{enumerate}
        \item Pick an initial point $\theta_0\in\mathbb{R}^m$ in parameter space.
        \item For $t=0,1,\dots ,T-1$:
        \begin{enumerate}
            \item Compute a classical local approximation $\tilde{f}_t$ of $f$ around $\theta_t$ (this typically involves the execution of various quantum circuits).
            \item Execute one (or several) classical optimization step(s) with respect to $\tilde{f}_t$ and obtain a point $\theta_{t+1}\in\mathbb{R}^m$ in parameter space.
        \end{enumerate}
        \item Output $\theta_T$.
    \end{enumerate}
    Instead of executing a fixed number of iterations $T$, one can, of course, impose some other stopping criterion. We remark that gradient descent can be seen as a special case of this technique, since computing the value of the function and its gradient at a point in parameter space is equivalent to computing the best local linear approximation at this point. Likewise, quantum analytic descent is a special case of this technique: The local approximation computed during quantum analytic descent is a trigonometric polynomial that coincides with $f$ up to second order at the given point in parameter space.
    
    The algorithm we will present in Section \ref{subsec:alg_description} also follows the above recipe; in light of this, we only need to describe how to obtain local approximations of $f$ around points in parameter space.

    \subsection{Description of the algorithm}
    \label{subsec:alg_description}
    Our algorithm employs the technique introduced in Section \ref{subsec:alg_strategy}. Accordingly, it suffices to describe how we construct local approximations of $f$ around points in parameter space. To this end, let $p\in\mathbb{R}^m$ be a fixed point in parameter space. We introduce a hyperparameter $L\in\mathbb{Z}$, where $1\leq L \leq m$, which controls the order of the approximation.
    We then determine all points $q_1 ,\dots , q_D\in \left\{-\frac{2\pi}{3},0,\frac{2\pi}{3}\right\}^m$ with the property that at most $L$ entries are non-zero (without repetitions, i.e., $q_i\neq q_j$ whenever $i\neq j$). Now, for $j=1,\dots ,D$, we set $p_j := p + q_j\in\mathbb{R}^m$. Following Lemma \ref{lemma:general_approx}, we then obtain an approximation $\tilde{f}\colon\mathbb{R}^m\to\mathbb{R}$ to $f$ by setting $\tilde{f}(\theta)= \sum_{j=1}^{D} \tilde\eta_j \tilde K(p_j, \theta)$, where $\tilde\eta\in\mathbb{R}^D$ is a solution to the linear system of equations
    \begin{align*}
				\left(\tilde K(p_i , p_j)\right)_{1\leq i,j\leq D}
				\cdot\tilde\eta
				=(f(p_1) ,\dots , f(p_D))^\intercal .
    \end{align*}
    At first glance it seems like the need to solve this linear system of equations entails a large classical computational overhead. However, the points $p_1 , \dots , p_D$ were specifically chosen so that the occurring matrix is the identity matrix, see Lemma \ref{lemma:orthonormal_basis}. It follows that $\tilde\eta=(f(p_1) ,\dots , f(p_D))^\intercal$, so that $\tilde{f}\colon\mathbb{R}^m\to\mathbb{R}$ is given by
    \begin{align}\label{eq:eval-formula}
        \tilde{f}(\theta)= \sum_{j=1}^{D} f(p_j ) \tilde K(p_j, \theta)
        =\sum_{j=1}^{D} f(p + q_j ) \tilde K(q_j, \theta - p)
        .
    \end{align}
    In order to obtain this approximation $\tilde{f}$, we need a quantum device to evaluate $f$ in $p+q_1 , \dots p+q_D$. Note that due to the exact formula \eqref{eq:eval-formula} this approach entails no classical computational overhead. The following result guarantees that $\tilde{f}$ is indeed a good local approximation of $f$ around $p$.

    \begin{theorem}\label{theorem:good_approx}
        Let $p\in\mathbb{R}^m$, and let $L\leq m$ be a positive integer. Moreover, let $q_1 ,\dots , q_D$ be the pairwise distinct elements of $\left\{-\frac{2\pi}{3},0,\frac{2\pi}{3}\right\}^m$ with the property that at most $L$ entries are non-zero. Now, define $\tilde{f}\colon\mathbb{R}^m\to\mathbb{R}$ by
    \begin{align*}
        \tilde{f}(\theta)
        =\sum_{j=1}^{D} f(p + q_j ) \tilde K(q_j, \theta - p)
        \text{ for all }\theta\in\mathbb{R}^m.
    \end{align*}
    Then the following holds:
    \begin{enumerate}[(i)]
        \item\label{theorem:good_approx_number_of_evaluations} In order to obtain the approximation $\tilde f$, one needs to evaluate $f$ in $D=\sum_{k=0}^L 2^k\cdot \binom{m}{k}$ points in parameter space.
        \item\label{theorem:good_approx_concides_on_subspaces} If $\vartheta\in\mathbb{R}^m$ has at most $L$ non-zero entries, then $\tilde{f}(p+\vartheta )=f(p+\vartheta )$. That is, 
        $\vartheta\mapsto\tilde{f}(p+\vartheta )$ coincides with $\vartheta\mapsto f(p+\vartheta )$ on any subspace of $\mathbb{R}^m$ that is spanned by at most $L$ of the coordinate axes.
        \item\label{theorem:good_approx_coincides_up_to_order_L} $\tilde f$ coincides with $f$ up to order $L$ locally around $p$. More precisely, we have $D^{\alpha}(f-\tilde{f})(p) = 0$ for all multiindices $\alpha\in (\mathbb{Z}_{\geq 0})^m$ with $|\alpha |\leq L$.
        \item\label{theorem:good_approx_error_estimate} Denote the union of subspaces of $\mathbb{R}^m$ that are spanned by at most $L$ of the coordinate axes as $\mathcal{C}_L$, i.e.,
        \begin{align*}
            \mathcal{C}_L = \bigcup_{\substack{S\subseteq\{\mathbf{e}_1 , \dots , \mathbf{e}_m\}\colon\\ \operatorname{card}(S)\leq L}} \operatorname{span}_{\mathbb{R}} (S) \subseteq\mathbb{R}^m,
        \end{align*}
        where $(\mathbf{e}_1 , \dots , \mathbf{e}_m )$ denotes the canonical basis of $\mathbb{R}^m$. Given any norm $\Vert\cdot\Vert$ on $\mathbb{R}^m$, the local approximation $\tilde{f}$ satisfies the estimate
        \begin{align*}
            |f(p+\vartheta ) - \tilde{f}(p+\vartheta)| = O(\Vert\vartheta\Vert^L\cdot\operatorname{dist}_{\Vert\cdot\Vert}(\vartheta , \mathcal{C}_L))
        \end{align*}
        as $\mathbb{R}^m\ni\vartheta\xrightarrow{}0$.
    \end{enumerate}
    \end{theorem}
    \begin{proof}
        Points (\ref{theorem:good_approx_number_of_evaluations}), (\ref{theorem:good_approx_concides_on_subspaces}), and (\ref{theorem:good_approx_coincides_up_to_order_L}) follow analogously to the proof of Theorem 7 in \cite{simon2024interpolating}. Point (\ref{theorem:good_approx_error_estimate}) follows from Lemma \ref{lemma:error_estimate_when_vanishing_along_certain_subspaces} in the appendix, applied to the function $\mathbb{R}^m\to\mathbb{R}$, $\vartheta\mapsto f(p+\vartheta)-\tilde{f}(p+\vartheta)$.
    \end{proof}

    A compact description of our algorithm in pseudocode, bringing together the technique outlined in Section \ref{subsec:alg_strategy} with the specific choice of local approximation described in Section \ref{subsec:alg_description}, can be found in Algorithm \ref{algorithm}.

    \begin{algorithm}[H]
		\SetKwInOut{Input}{Input}
		\SetKwInOut{Output}{Output}
		
		\Input{order $L$, number of iterations $T\in\mathbb{Z}_{\geq 1}$, function $f\colon\mathbb{R}^m\to\mathbb{R}$, $f(\theta) = \langle\psi (\theta)|\mathcal{M}|\psi (\theta)\rangle$ (see Section \ref{subsec:alg_setting})}
		\Output{point $\theta_{T}\in\mathbb{R}^m$ in parameter space}

        \nl Determine the pairwise distinct elements $q_1 ,\dots , q_D$ of $\left\{-\frac{2\pi}{3},0,\frac{2\pi}{3}\right\}^m$ with the property that at most $L$ entries are non-zero

        \nl Pick an initial point $\theta_0\in\mathbb{R}^m$ in parameter space
  
		\For{$t=0,1,\dots ,T-1$}
		{
            \For{$j=1,2,\dots ,D$}
            {
                \nl Evaluate $f(\theta_t + q_j)$ using a quantum device
            }

            \nl Set $\tilde{f}_t(\theta)=\sum_{j=1}^{D} f(\theta_t + q_j ) \tilde K(q_j, \theta - \theta_t)$

            \nl Execute one (or several) classical optimization step(s) with respect to $\tilde{f}_t$ and obtain a point $\theta_{t+1}\in\mathbb{R}^m$ in parameter space
		}
  
		\nl \Return $\theta_{T}$
		
		\caption{\algoCap}
		\label{algorithm}
	\end{algorithm}

    In Algorithm \ref{algorithm}, the method for minimizing the classical approximations $\tilde{f}_t$ is intentionally left open in order to not limit the generality of the algorithm. A natural approach would be to carry out several gradient descent steps with respect to $\tilde{f}_t$, whose gradient can be computed efficiently.
    


    \subsection{Comparison with existing approaches}
    \label{subsec:alg_comparison}

    Since all algorithms investigated below are based on computing and subsequently optimizing local approximations (see also Section \ref{subsec:alg_strategy}) we limit ourselves to comparing the quality of the latter.
    It follows from Theorem \ref{theorem:good_approx}(\ref{theorem:good_approx_number_of_evaluations}) that the number of evaluations of $f$ needed by Algorithm \ref{algorithm} to compute the local approximation $\tilde{f}_t$ does not exceed (but is comparable to) the number of evaluations necessary to compute the $L$-th order Taylor polynomial about $\theta_t$ using the parameter-shift rules for arbitrary-order partial derivatives \cite{Mari_2021}, \cite{Wierichs2022generalparameter}, \cite{simon2024interpolating}. Moreover, Theorem \ref{theorem:good_approx}(\ref{theorem:good_approx_coincides_up_to_order_L}) guarantees that $\tilde{f}_t$ coincides with $f$ up to order $L$ locally around $\theta_t$, and the same holds true for the $L$-th order Taylor polynomial about $\theta_t$.
    However, in contrast to the $L$-th order Taylor polynomial about $\theta_t$, the {\emph{local}} approximation $\tilde{f}_t$ is guaranteed to preserve some {\emph{global}} properties of $f$, such as the periodicity of $f$ (since $\tilde{f}_t\in H$) and the values of $f$ along certain $L$-dimensional affine subspaces through $\theta_t$ (by Theorem \ref{theorem:good_approx}(\ref{theorem:good_approx_concides_on_subspaces})).
    A heuristic explanation for these advantages is that Algorithm \ref{algorithm} exploits the intrinsic structure of $H$, which is a finite-dimensional subspace of the much larger infinite-dimensional space of smooth $\mathbb{R}$-valued functions on $\mathbb{R}^m$, to which Taylor approximations generically apply.
    Moreover, in terms of the displacement $\vartheta$ from the development point, $L$-th order Taylor approximation comes with an error estimate of the form $O(\Vert\vartheta\Vert^{L+1})$, whereas the local approximations featuring in Algorithm \ref{algorithm} satisfy a better error estimate of the form $O(\Vert\vartheta\Vert^L\cdot\operatorname{dist}_{\Vert\cdot\Vert}(\vartheta , \mathcal{C}_L))$, see Theorem \ref{theorem:good_approx}(\ref{theorem:good_approx_error_estimate}).
    Since the approximation error of the $L$-th order Taylor approximation can be shown to not satisfy the latter estimate in general (even when restricting attention to functions in $H$), the error estimate in Theorem \ref{theorem:good_approx}(\ref{theorem:good_approx_error_estimate}) is in fact \emph{strictly} better than the corresponding error estimate for $L$-th order Taylor approximation.
    In the experiments in \cite{simon2024interpolating} the approximations featuring in Algorithm \ref{algorithm}, albeit with $\pm\tfrac{\pi}{2}$ parameter-shifts (instead of $\pm \tfrac{2\pi}{3}$ shifts in the present work), have outperformed Taylor approximation for $L = 1,2,3,4$.   
    Below, we will compare Algorithm \ref{algorithm} with gradient descent for $L=1$ and quantum analytic descent for $L=2$. The technical reasons for this specific choice of comparisons will be explained in the respective sections.

    \subsubsection{Comparison with gradient descent}\label{subsubsec:comparison_with_gradient_descent}

    For functions $f\colon\mathbb{R}^m\to\mathbb{R}$ of the form \eqref{eq:f}, gradient descent using the parameter-shift rules requires $2m$ circuit evaluations per iteration. Keeping track of the value of $f$ during the execution of gradient descent requires one additional circuit evaluation per iteration for a total of $2m+1$ evaluations, since computation of the gradient via the parameter-shift rules does not require knowledge of the value of $f$ at the development point.
    With the choice of hyperparameter $L=1$, Algorithm \ref{algorithm} requires $2m+1$ circuit evaluations per iteration by Theorem \ref{theorem:good_approx}(\ref{theorem:good_approx_number_of_evaluations}). Consequently, it makes sense to compare gradient descent to Algorithm \ref{algorithm} with $L=1$, because in this case both algorithms require the same number of circuit evaluations per iteration.
    
    A single iteration of both gradient descent and Algorithm \ref{algorithm} with $L=1$ corresponds to an optimization step with respect to an approximation that coincides with $f$ up to first order around the development point. However, the approximations of $f$ computed during the execution of Algorithm \ref{algorithm} are elements of the reproducing kernel Hilbert space $H$ (and consequently trigonometric polynomials), whereas gradient descent corresponds to optimization steps with respect to linear approximations of $f$. Moreover, it follows from Theorem \ref{theorem:good_approx}(\ref{theorem:good_approx_concides_on_subspaces}) that the approximations of $f$ computed during the execution of Algorithm \ref{algorithm} with $L=1$ coincide with $f$ on the (affine) lines through the respective development point which are parallel to the coordinate axes of $\mathbb{R}^m$. Informally speaking, this means that if we move away from the development point along a single coordinate direction, the approximation is exact, i.e., it coincides with $f$. Consequently, there are, in general, regions in parameter space arbitrarily close to the development point where the approximation computed during execution of Algorithm \ref{algorithm} with $L=1$ significantly outperforms the first order Taylor polynomial (which is the linear approximation relevant in the context of gradient descent). This fact is expressed by the estimate in Theorem \ref{theorem:good_approx}(\ref{theorem:good_approx_error_estimate}) which, as mentioned above, is strictly better than the corresponding estimate for Taylor approximation. In Section \ref{sec:experiments} of the present work we carry out extensive experiments where we do not only compare the quality of the local approximations, but also the performance of gradient descent and Algorithm \ref{algorithm} with $L=1$ when it comes to the actual goal of minimizing the objective function $f$.

    \subsubsection{Comparison with quantum analytic descent}\label{subsubsec:comparison_with_quantum_analytic_descent}
    
    Quantum analytic descent requires $2m^2+m+1$ circuit evaluations per iteration. With the choice of hyperparameter $L=2$, Algorithm \ref{algorithm} requires $2m^2 +1$ circuit evaluations per iteration by Theorem \ref{theorem:good_approx}(\ref{theorem:good_approx_number_of_evaluations}). So, Algorithm \ref{algorithm} with $L=2$ requires slightly fewer circuit evaluations than quantum analytic descent, but the number of circuit evaluations is the same asymptotically. Consequently, it makes sense to compare quantum analytic descent to Algorithm \ref{algorithm} with $L=2$.

    A single iteration of both quantum analytic descent and Algorithm \ref{algorithm} with $L=2$ corresponds to an optimization step with respect to an approximation that coincides with $f$ up to second order around the development point. The approximations of $f$ computed during execution of both algorithms are trigonometric polynomials that preserve the periodicity of $f$. Informally speaking, if we move away from the development point along {\emph{a single}} coordinate direction, the approximation computed by quantum analytic descent is exact, i.e., it coincides with $f$. However, for the approximations computed during execution of Algorithm \ref{algorithm} with $L=2$ we have the following stronger guarantee (by Theorem \ref{theorem:good_approx}(\ref{theorem:good_approx_concides_on_subspaces})): If we move away from the development point along {\emph{at most two}} coordinate directions simultaneously, the approximation is exact, i.e., it coincides with $f$.
    In fact, in terms of the displacement $\vartheta$ from the development point, the approximations computed during execution of quantum analytic descent can be shown to satisfy an error estimate of the form $O(\Vert\vartheta\Vert^2\cdot\operatorname{dist}_{\Vert\cdot\Vert}(\vartheta , \mathcal{C}_1))$, but, in general, \emph{not} the strictly stronger error estimate of the form $O(\Vert\vartheta\Vert^2\cdot\operatorname{dist}_{\Vert\cdot\Vert}(\vartheta , \mathcal{C}_2))$, which is satisfied by the local approximations computed during execution of Algorithm \ref{algorithm} with $L=2$, see Theorem \ref{theorem:good_approx}(\ref{theorem:good_approx_error_estimate}).
    In Section \ref{sec:experiments} of the present work we carry out extensive experiments where we do not only compare the quality of the local approximations, but also the performance of quantum analytic descent and Algorithm \ref{algorithm} with $L=2$ when it comes to the actual goal of minimizing $f$.

    \section{EXPERIMENTS}
	\label{sec:experiments}

In this section, we present a detailed experimental comparison of kernel descent with 
$L=1$ versus gradient descent, as well as kernel descent with 
$L=2$ versus quantum analytic descent.
As explained in Section \ref{subsec:alg_comparison}, these pairings are chosen because of their roughly equivalent computational efforts in terms of circuit evaluations per iteration.
The experiments are conducted in two main parts. The first part, described in Sections \ref{sec:sampling}, \ref{subsec:experiments_approx_quality}, and \ref{susbsec:experiments_performance}, is concerned with randomized circuits and observables. This approach allows us to consider a large, \emph{statistically significant} number of such circuits and observables.
The first part is structured as follows:

The algorithms will be compared with respect to the quality of the local approximations (Section \ref{subsec:experiments_approx_quality}) and their performance in terms of their ability to minimize the objective function (Section \ref{susbsec:experiments_performance}). 
Before that, we describe the experimental setup in Section \ref{sec:sampling}.

In the second part, which is comprised of Section \ref{subsec:experiments_spin_ring_hamiltonian}, we deal with a problem from the literature on VQAs.

\subsection{Setup for experiments with randomized circuits and observables}\label{sec:sampling}

Here we describe the setup for our experiments in Sections \ref{subsec:experiments_approx_quality} and \ref{susbsec:experiments_performance}. In these experiments we neither take measurement shot noise nor quantum hardware noise into account. That is, all compared algorithms are provided with exact values of the objective function $f$ computed via statevector simulation.

In order to ensure independence of our results from any particular choice of circuit or configuration, we repeatedly sample circuits, observables, and points in parameter space randomly in all our experiments. We proceed similarly to the experimental setup in \cite{simon2024denoising}: Points in parameter space are randomly sampled from the uniform distribution on $[-\pi , \pi )^m$. With the notation from Section \ref{subsec:alg_setting}, the observable $\mathcal{M}$ is an $n$-qubit Pauli randomly sampled from the uniform distribution on $\{I,X,Y,Z\}^{\otimes n}\setminus\{I^{\otimes n}\}$. The generators $G_1 , \dots , G_m$ of the parametrized gates $R_1 , \dots , R_m$ are sampled (independently) from the uniform distribution on $\{I,X,Y,Z\}^{\otimes n}\setminus\{I^{\otimes n}\}$, and the $n$-qubit unitaries $C_1 , \dots , C_{m+1}$ are randomly sampled (independently) via the same procedure by which the individual layers in the well-known quantum volume test \cite{quantum_volume} are sampled. For a visual representation of the entire circuit we refer to Figure \ref{fig:general_circuit}, the specific choice of the unitaries $C_j$ is explained and visualized in Figure \ref{fig:qv_layer_circuit}.

    \begin{FigureInColumn}
        \centering
        {\begin{quantikz}
                &[2mm] \gate[][0.7cm][2cm]{C_j}\qwbundle{n} & \qw
            \end{quantikz}} = \resizebox{0.5\columnwidth}{!}{
            \begin{quantikz}
                & \gate[5, nwires=3]{\tau_j}   & \gate[2, nwires=3]{U^{(j)}_1 \in SU(4)} &\qw\\[-10pt]
                & \qw                 & \qw & \qw \\[-10pt]
                & \lstick{$\vdots$}            & \lstick{$\vdots$} & &  \\[-3pt]
                & \qw                 &\gate[2, nwires=3]{U^{(j)}_{\left \lfloor{n/2}\right \rfloor} \in SU(4)}   &  \qw \\[-10pt]
                & \qw                & &  \qw      
            \end{quantikz}}    
        
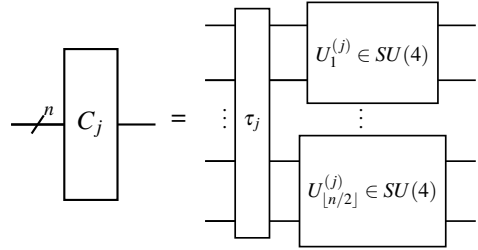
\captionof{figure}{This figure shows the specific choice of unitaries $C_j$ featuring in the experiments in Sections \ref{subsec:experiments_approx_quality} and \ref{susbsec:experiments_performance}: First, a permutation $\tau_j\in S_n$ is sampled uniformly at random. Subsequently, unitaries $U^{(j)}_1 ,\dots , U^{(j)}_{\left \lfloor{n/2}\right \rfloor }$ are sampled independently from the Haar measure on $SU(4)$. These are then applied to the qubit pairs $(\tau_j (1),\tau_j (2)),\dots ,(\tau_j (2{\left \lfloor{n/2}\right \rfloor }-1),\tau_j (2{\left \lfloor{n/2}\right \rfloor }))$ respectively. Note that this is precisely how the individual layers in the quantum volume test \cite{quantum_volume} are sampled.}
        \label{fig:qv_layer_circuit}
    \end{FigureInColumn}

\subsection{Quality of the local approximations}
\label{subsec:experiments_approx_quality}

Let $p\in\mathbb{R}^m$ be a point in parameter space, let $g\colon\mathbb{R}^m\to\mathbb{R}$ be a local approximation of $f$ about $p$.
    In this subsection we will provide comparisons of the two algorithm pairs kernel descent with $L=1$ versus gradient descent and kernel descent with $L=2$ versus quantum analytic descent with respect to three notions of {\emph{error}} at a point $\theta\in\mathbb{R}^m$ that is \enquote{sufficiently close} to the development point $p$:

     \begin{itemize}
        \item The {\emph{value approximation error}} $$|f(\theta)-g(\theta)|\,,$$ which is the absolute value of the difference between the value of the true function and that of the local approximation.
        \item The {\emph{gradient approximation error in Euclidean norm}} 
        $$\Vert \nabla{f}(\theta)-\nabla{g}(\theta)\Vert\,,$$
        which is the Euclidean norm of the difference between the gradient of the true function and that of the local approximation (recall that this is relevant, since all compared algorithms -- either explicitly or implicitly -- apply gradient descent to the respective local approximations in an inner minimization loop). 
        \item The {\emph{gradient approximation error in cosine distance}}. That is, the cosine distance between the gradient of the true function and that of the local approximation
        $$1-\frac{\nabla{f}(\theta)\bullet \nabla{g}(\theta)}{\Vert \nabla{f}(\theta)\Vert \Vert \nabla{g}(\theta)\Vert},$$
        where $\bullet$ denotes the Euclidean dot product.
        The cosine distance allows us to only compare the error in the {\emph{direction}} of the gradient of the approximation, ignoring the error in the {\emph{length}}. This is of particular relevance for the experiments in Section \ref{subsubsec:experiments_kd_vs_gd}.
        (In order to prevent numerical instability and division by $0$ in our experiments, a small positive constant was added to both factors appearing in the denominator.)
    \end{itemize}

These three measures lead to six sets of comparisons (three comparisons per one of the two algorithm pairs), which will each be visualized via two kinds of scatter plots that highlight different aspects of the comparison. Since the process of sampling, error computation, plotting is the same for all six experiments, we will describe the procedure once for the first experiment, kernel descent with $L=1$ versus gradient descent with respect to the value approximation error. The others work analogously, just replace $L=1$ with $L=2$ and gradient descent with quantum analytic descent, and/or exchange the value approximation error for any of the other two error measures. The results of the experiments are described in Section \ref{subsubsec:experiments_approx_capabilities_results}.   

\paragraph{Sampling.}
We denote the value approximation error at $\theta$ by $\mathcal{E}_{\theta}(f, g)$.
Now fix $n=m=10$ and repeat the following $N=25,000$ times:

 \begin{enumerate}
        \item A circuit, a point $p\in\mathbb{R}^m$ and an observable $\mathcal{M}$ are sampled randomly as explained in Section \ref{sec:sampling}, and $f$ is defined as in Section \ref{subsec:alg_setting}.
        \item Determine the local approximations $f_{\operatorname{GD}}$ and $f_{\operatorname{KD}}$ to $f$ about $p$ computed by gradient descent and kernel descent with $L=1$, respectively. 
        \item Randomly sample a point $\theta$ in the vicinity of $p$; specifically, we set $\theta = p+\vartheta$, where $\vartheta$ is sampled from the uniform distribution on $[-0.5,0.5)^m$.
        \item Compute the distance $d_{\theta}:=\Vert \theta -p\Vert = \Vert\vartheta\Vert$ to the development point $p$.
        \item Compute the errors $\mathcal{E}_{\theta}(f, f_{\operatorname{GD}})$ and $\mathcal{E}_{\theta}(f, f_{\operatorname{KD}})$.
    \end{enumerate}
This process yields $25,000$ data points of the form $[d_{\theta}, \mathcal{E}_{\theta}(f, f_{\operatorname{GD}}), \mathcal{E}_{\theta}(f, f_{\operatorname{KD}})]$ which we use to create the two aforementioned types of scatter plots.

\paragraph{The first type of scatter plot.} 
  The goal of the first type of scatter plot is to directly compare the errors made by gradient descent and kernel descent. 
  To this end, we plot the $N=25,000$ points of the form $[\mathcal{E}_{\theta}(f, f_{\operatorname{GD}}), \mathcal{E}_{\theta}(f, f_{\operatorname{KD}})]$ in a two-dimensional coordinate system. 
The $x$-coordinate represents the error made by gradient descent, while the $y$-coordinate represents the error made by kernel descent with $L=1$. 
As a consequence, points above the diagonal correspond to outcomes where gradient descent outperforms kernel descent, while points below the diagonal indicate the opposite. 

Points on or above the diagonal are marked in blue, and points below the diagonal are marked in red. Our measure of success for kernel descent is the percentage of points that lie below the diagonal. However, this does not quantify the margin by which one algorithm typically outperforms the other. To provide an indication of this margin, we plot a line through the origin, with its polar angle being the average of the polar angles of the $N=25,000$ plotted points. 
  
\paragraph{The second type of scatter plot.}     

The second type of scatter plot focuses on illustrating the dependence of the approximation errors on the distance to the development point $p$, which is not represented in the first type of plot. From each of the $N=25,000$ data points $[d_{\theta}, \mathcal{E}_{\theta}(f, f_{\operatorname{GD}}), \mathcal{E}_{\theta}(f, f_{\operatorname{KD}})]$, 
we derive two points: $[d_{\theta}, \mathcal{E}_{\theta}(f, f_{\operatorname{GD}})]$ and $[d_{\theta}, \mathcal{E}_{\theta}(f, f_{\operatorname{KD}})]$ and plot them in a two-dimensional coordinate system, marked in blue and red, respectively. This results in a scatter plot of $2N=50,000$ points.

Generally speaking, the closer a point is to the $x$-axis, the better the approximation performance of the corresponding algorithm. 
    As a visual aid, we fit a curve of the form $x\mapsto c \cdot x^2$ to each point cloud (minimizing the mean squared error), where $c$ is the parameter being fit. We chose the exponent $2$, since the local value approximation error of both algorithms is $O(\Vert\theta - p\Vert^{2})$, see Theorem \ref{theorem:good_approx}(\ref{theorem:good_approx_coincides_up_to_order_L}).
    
    For the five remaining combinations of algorithm pair and error notion, the exponent is likewise chosen according to the estimates in Theorem \ref{theorem:good_approx}(\ref{theorem:good_approx_coincides_up_to_order_L}).
    As a result of this construction, the curve that lies below the other indicates that the corresponding algorithm performs better in terms of approximation quality.

    \subsubsection{Results}
    \label{subsubsec:experiments_approx_capabilities_results}

    Based on the criteria described above, kernel descent outperformed in all twelve scatter plots.

    In the six scatter plots of the first type (see Figure \ref{fig:scatter_typeI}),  kernel descent won all comparisons, with at least $58\%$ of the points ---and often significantly more--- marked in red in each plot.

    In the six scatter plots of the second type (see Figure \ref{fig:scatter_typeII}), kernel descent won all comparisons, as the red curve lies below the blue curve in every plot. 
    
    \begin{figure*}[htbp]
        \captionsetup{width=1.0\linewidth}
    	\centering
        \subcaptionbox[width=0.01\linewidth]{Comparison between gradient descent and kernel descent with $L=1$ wrt.\ value approximation error. Kernel descent outperformed gradient descent in $63.7\%$ of the outcomes.}{\includegraphics[width=2.13in]{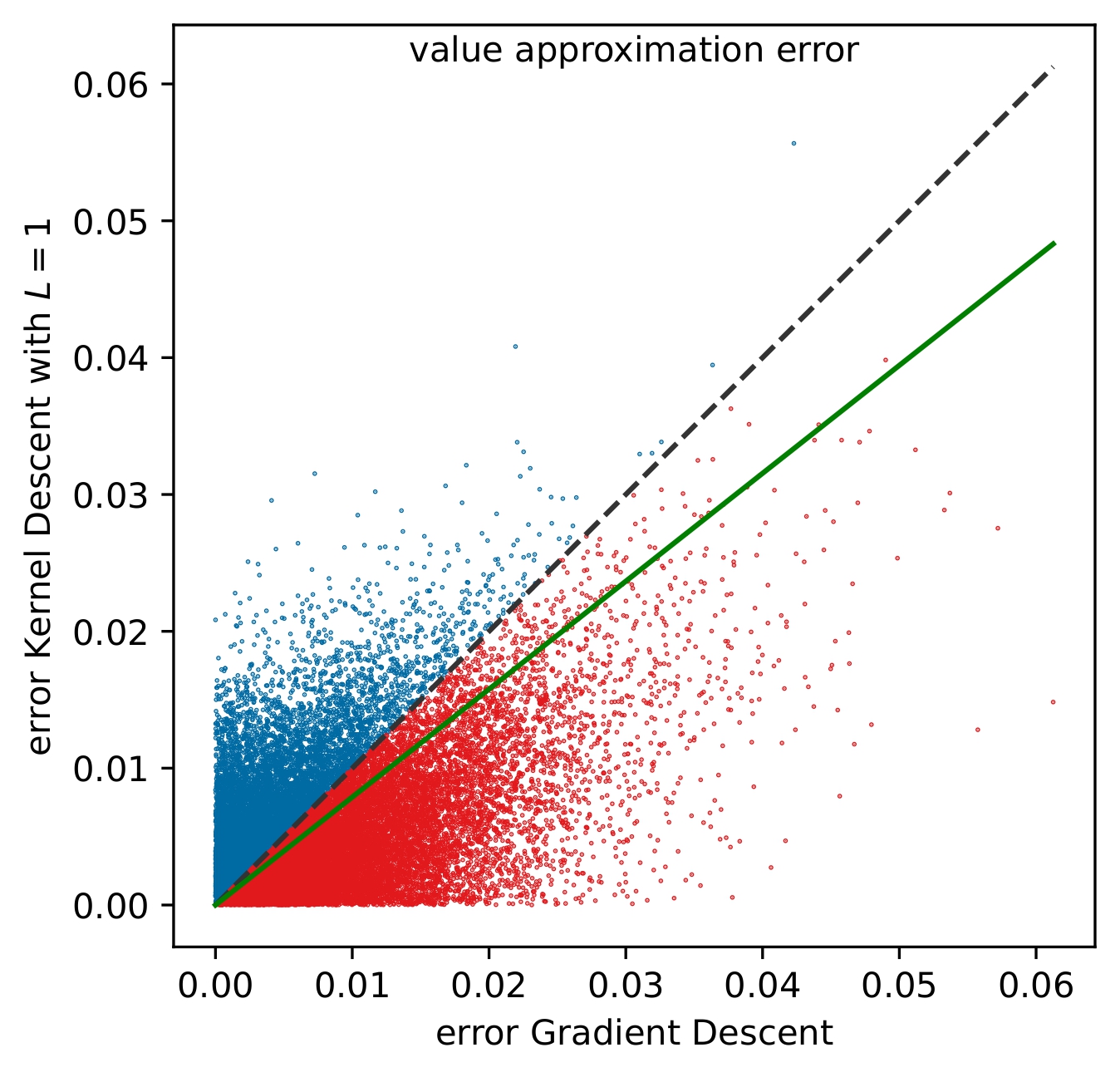}}\hspace{1.0em}%
        \subcaptionbox{Comparison between gradient descent and kernel descent with $L=1$ wrt.\ gradient approximation error in Euclidean norm. Kernel descent outperformed gradient descent in $76.2\%$ of the outcomes.}{\includegraphics[width=2.13in]{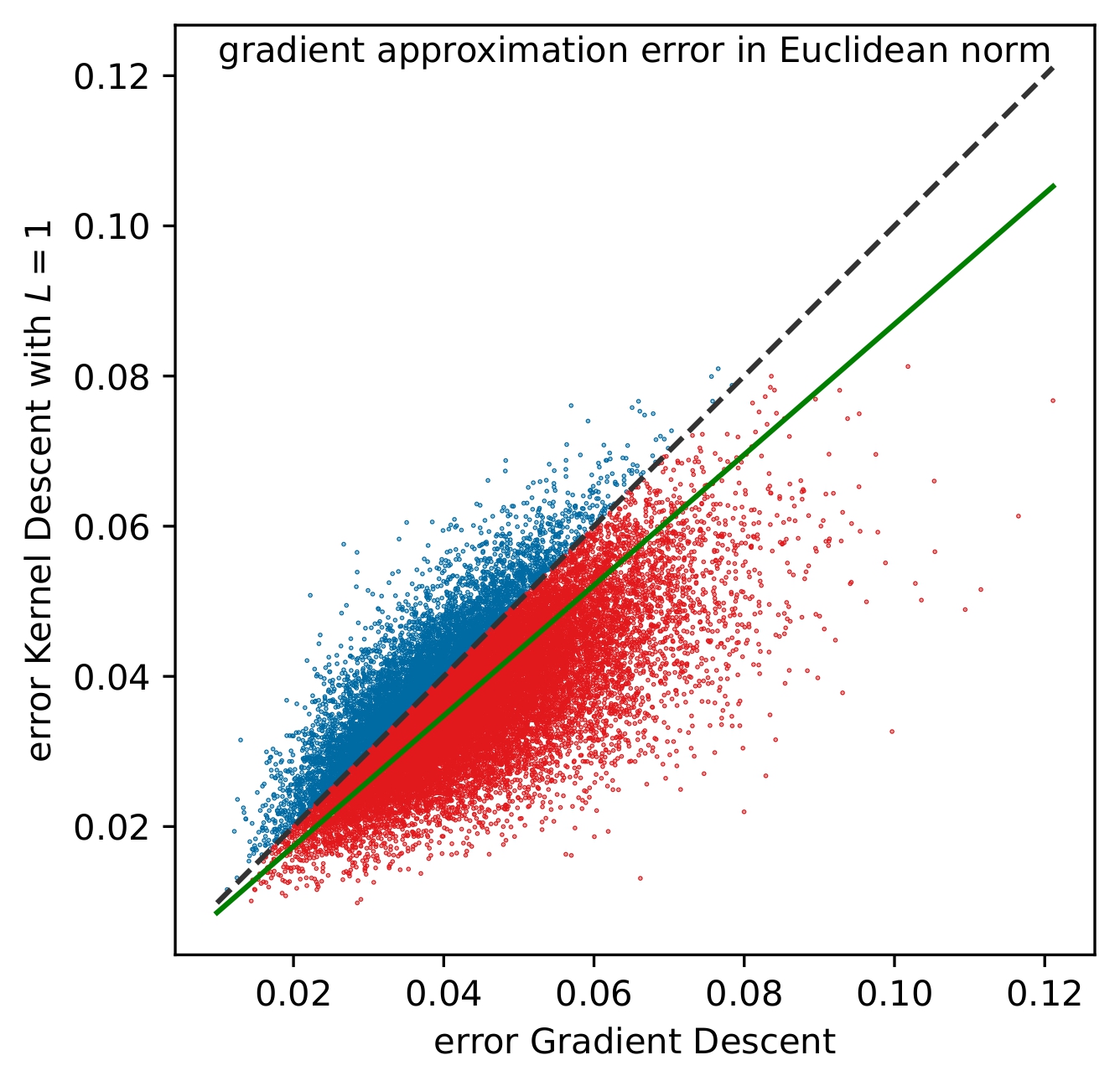}}\hspace{1.0em}%
        \subcaptionbox{Comparison between gradient descent and kernel descent with $L=1$ wrt.\ gradient approximation error in cosine distance. Kernel descent outperformed gradient descent in $71.3\%$ of the outcomes.}{\includegraphics[width=2.13in]{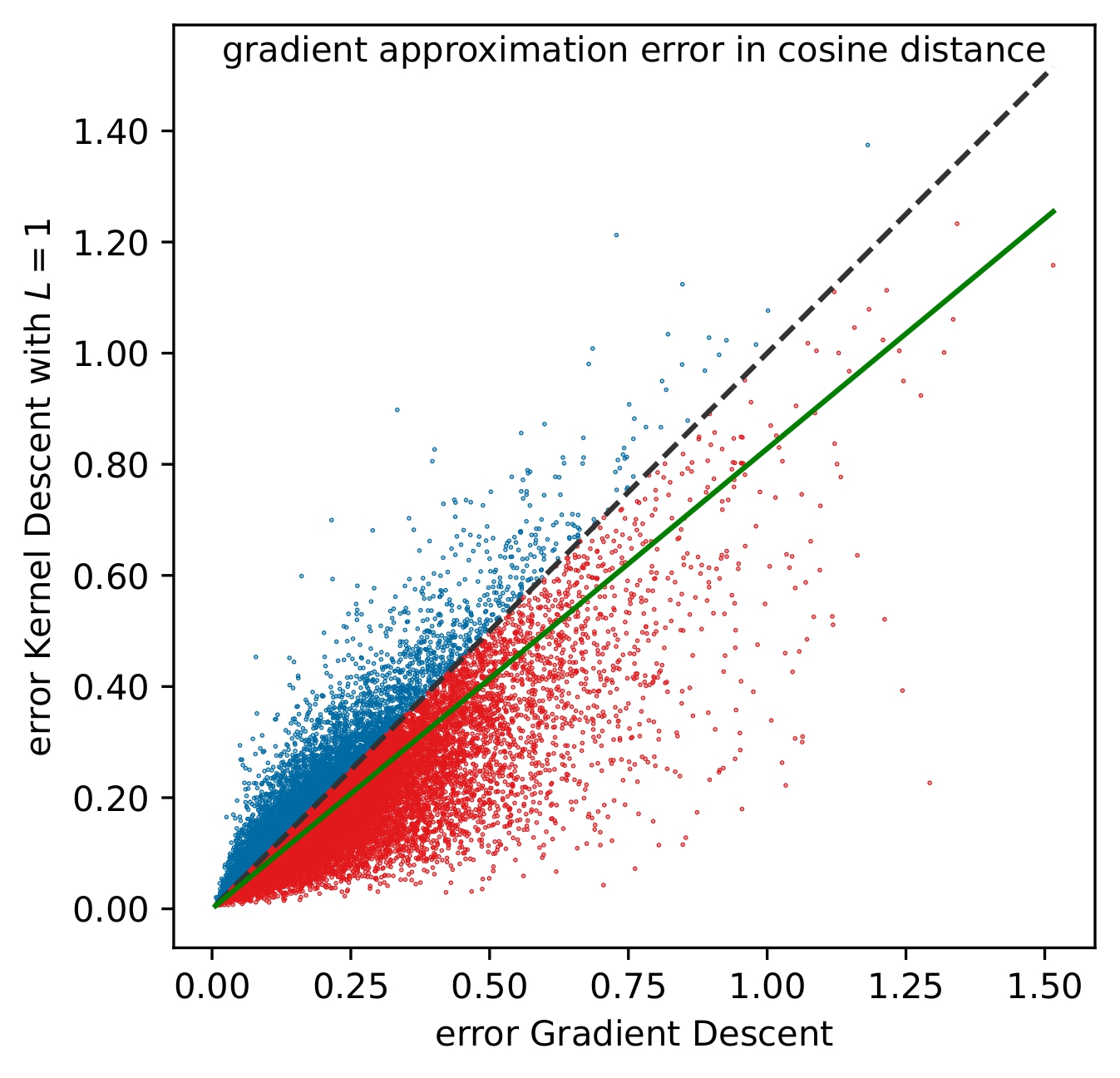}}
    	
    	\bigskip
    	
        \subcaptionbox{Comparison between quantum analytic descent and kernel descent with $L=2$ wrt.\ value approximation error. Kernel descent outperformed quantum analytic descent in $58.7\%$ of the outcomes.}{\includegraphics[width=2.13in]{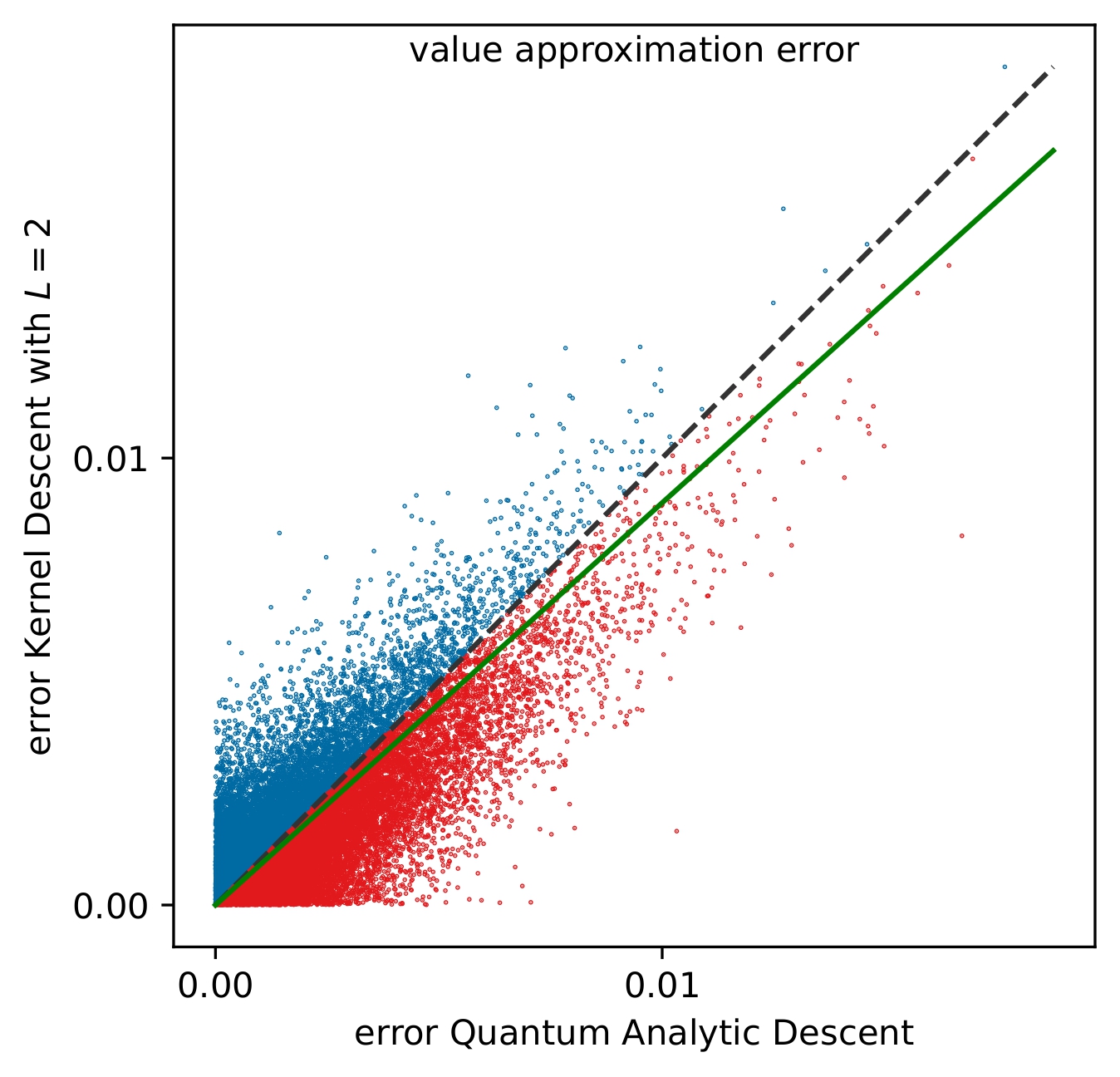}}\hspace{1.0em}%
        \subcaptionbox{Comparison between quantum analytic descent and kernel descent with $L=2$ wrt.\ gradient approximation error in Euclidean norm. Kernel descent outperformed quantum analytic descent in $79.9\%$ of the outcomes.}{\includegraphics[width=2.13in]{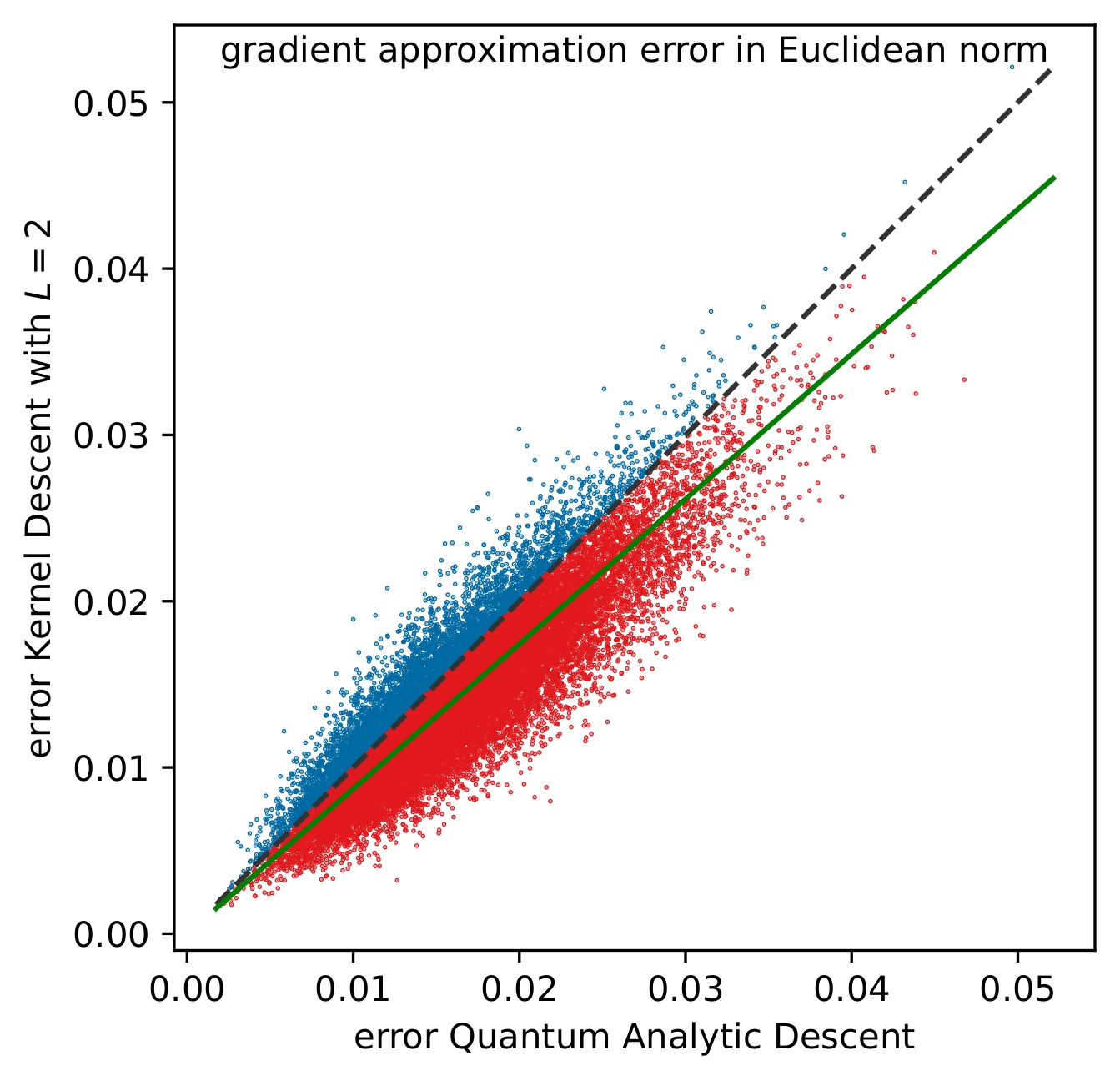}}\hspace{1.0em}%
        \subcaptionbox{Comparison between quantum analytic descent and kernel descent with $L=2$ wrt.\ gradient approximation error in cosine distance. Kernel descent outperformed quantum analytic descent in $78.5\%$ of the outcomes.}{\includegraphics[width=2.13in]{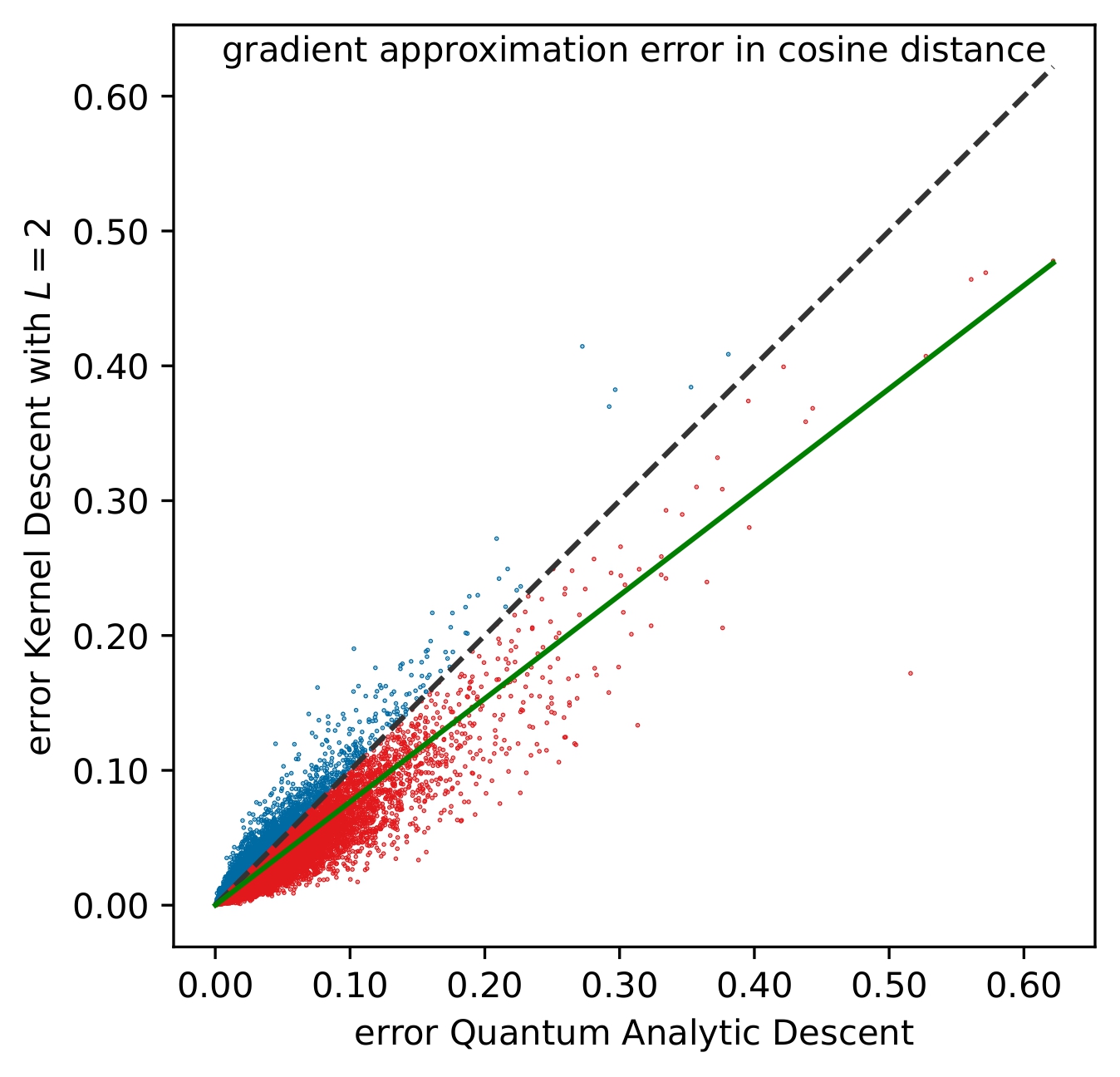}}
     
    	\caption{This figure shows the scatter plots of the first type from Section \ref{subsec:experiments_approx_quality}. Points below the diagonal (colored red) correspond to outcomes where kernel descent outperformed the algorithm it was being compared to. The green line is the line through the origin, whose polar angle is the average of the polar angles of the $N = 25,000$ points appearing in the scatter plot. The fact that the green lines are running below the diagonals is a further indication that kernel descent outperformed the algorithm it was being compared to; see Section \ref{subsec:experiments_approx_quality} for details.}
    	\label{fig:scatter_typeI}
    \end{figure*}

        \begin{figure*}[htbp]
        \captionsetup{width=1.0\linewidth}
    	\centering
        \subcaptionbox[width=0.01\linewidth]{Comparison between gradient descent and kernel descent with $L=1$ wrt.\ value approximation error. We fit curves of the form $x\mapsto c\cdot x^2$ to both point clouds.}{\includegraphics[width=2.13in]{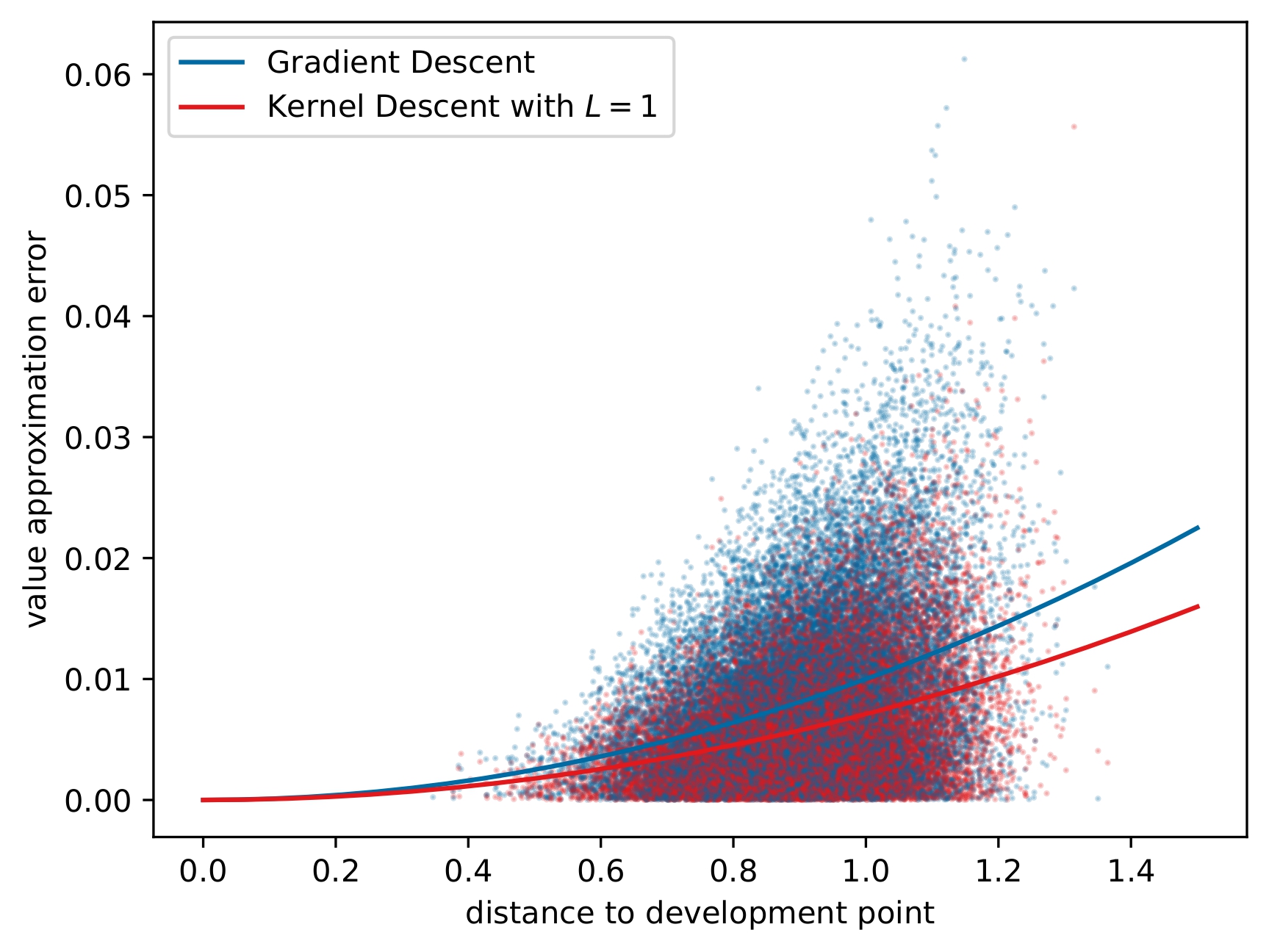}}\hspace{1.0em}%
        \subcaptionbox{Comparison between gradient descent and kernel descent with $L=1$ wrt.\ gradient approximation error in Euclidean norm. We fit curves of the form $x\mapsto c\cdot x^1$ to both point clouds.}{\includegraphics[width=2.13in]{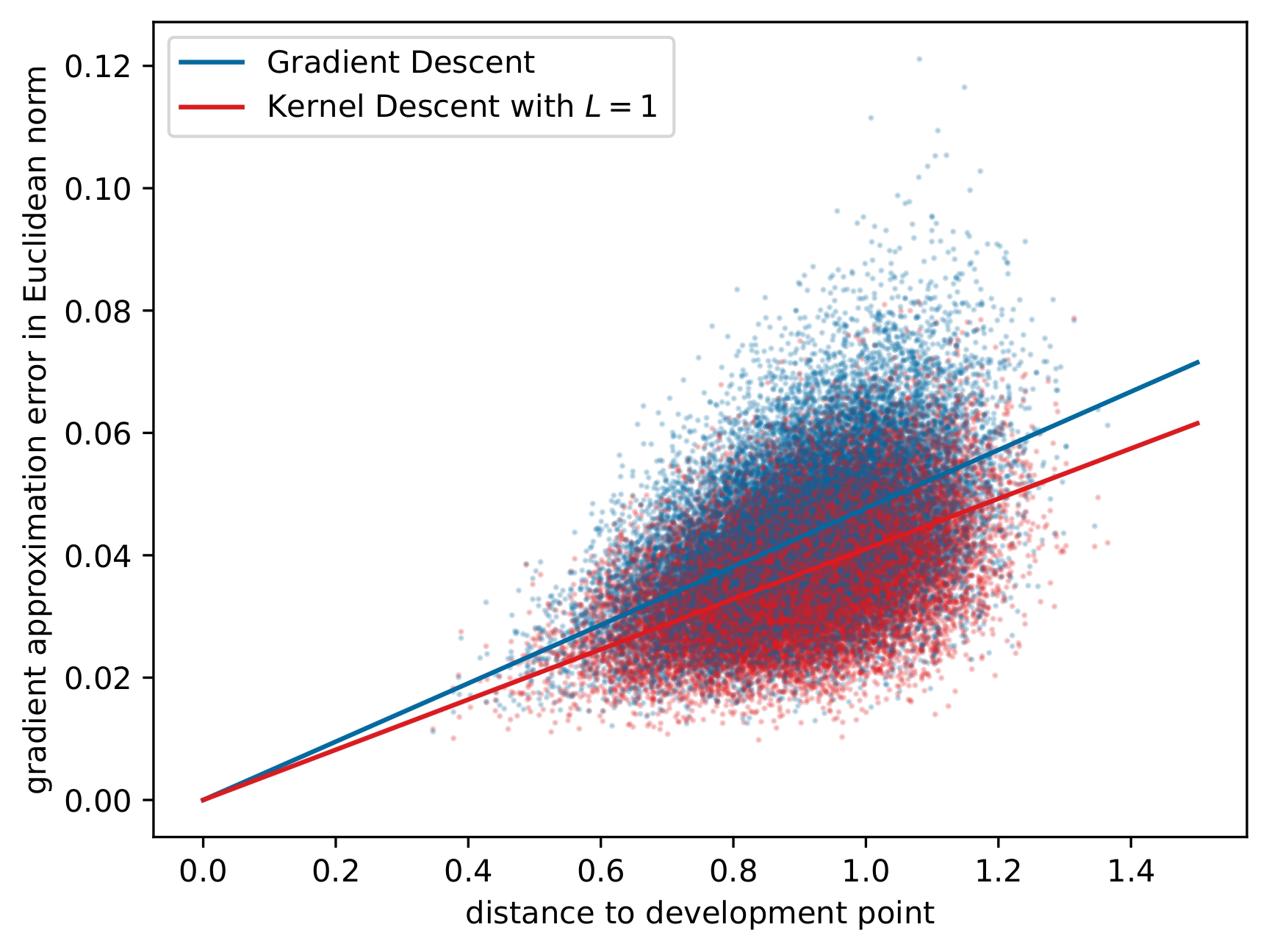}}\hspace{1.0em}%
        \subcaptionbox{Comparison between gradient descent and kernel descent with $L=1$ wrt.\ gradient approximation error in cosine distance. We fit curves of the form $x\mapsto c\cdot x^2$ to both point clouds.}{\includegraphics[width=2.13in]{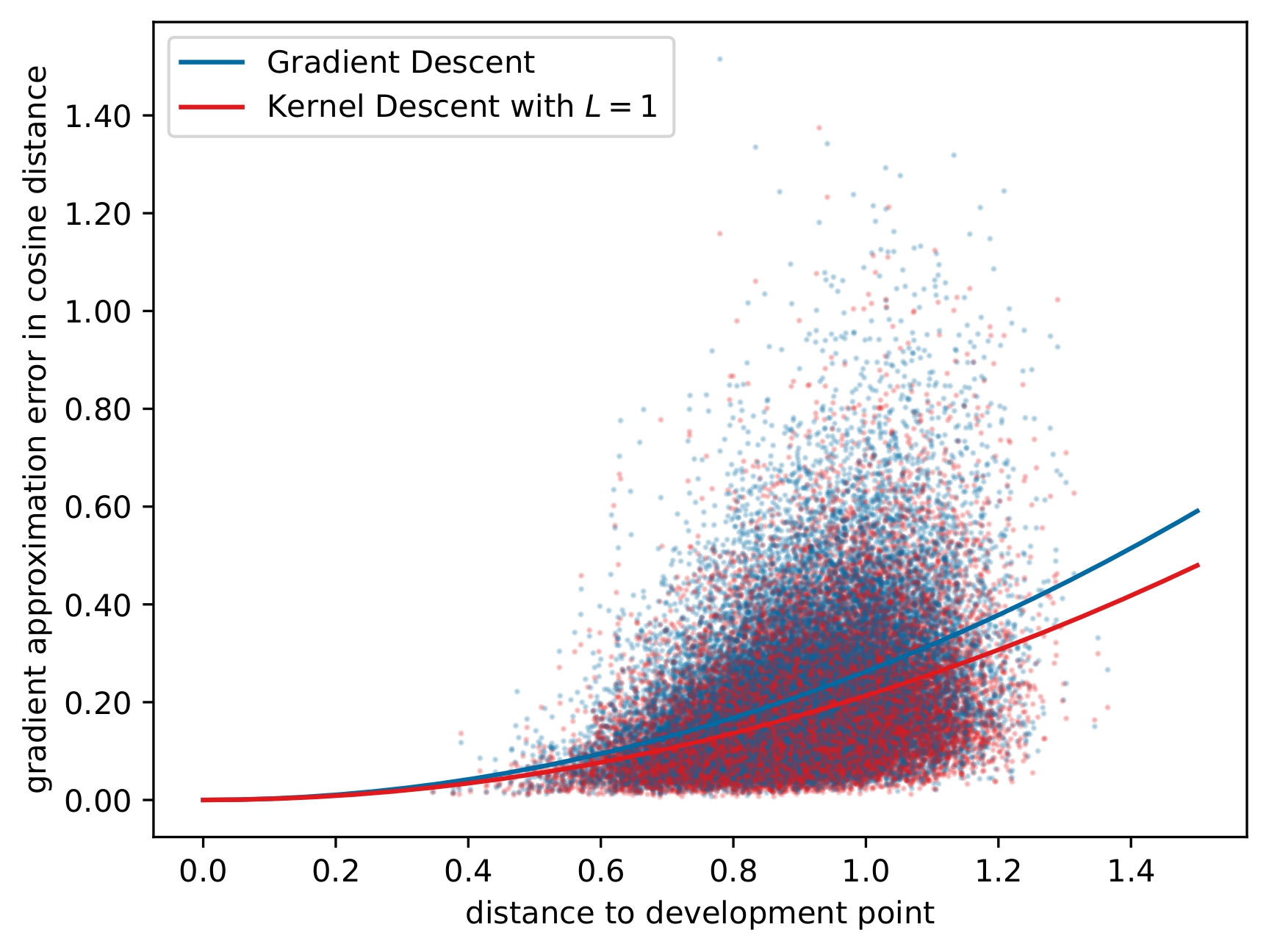}}
    	
    	\bigskip
    	
        \subcaptionbox{Comparison between quantum analytic descent and kernel descent with $L=2$ wrt.\ value approximation error. We fit curves of the form $x\mapsto c\cdot x^3$ to both point clouds.}{\includegraphics[width=2.13in]{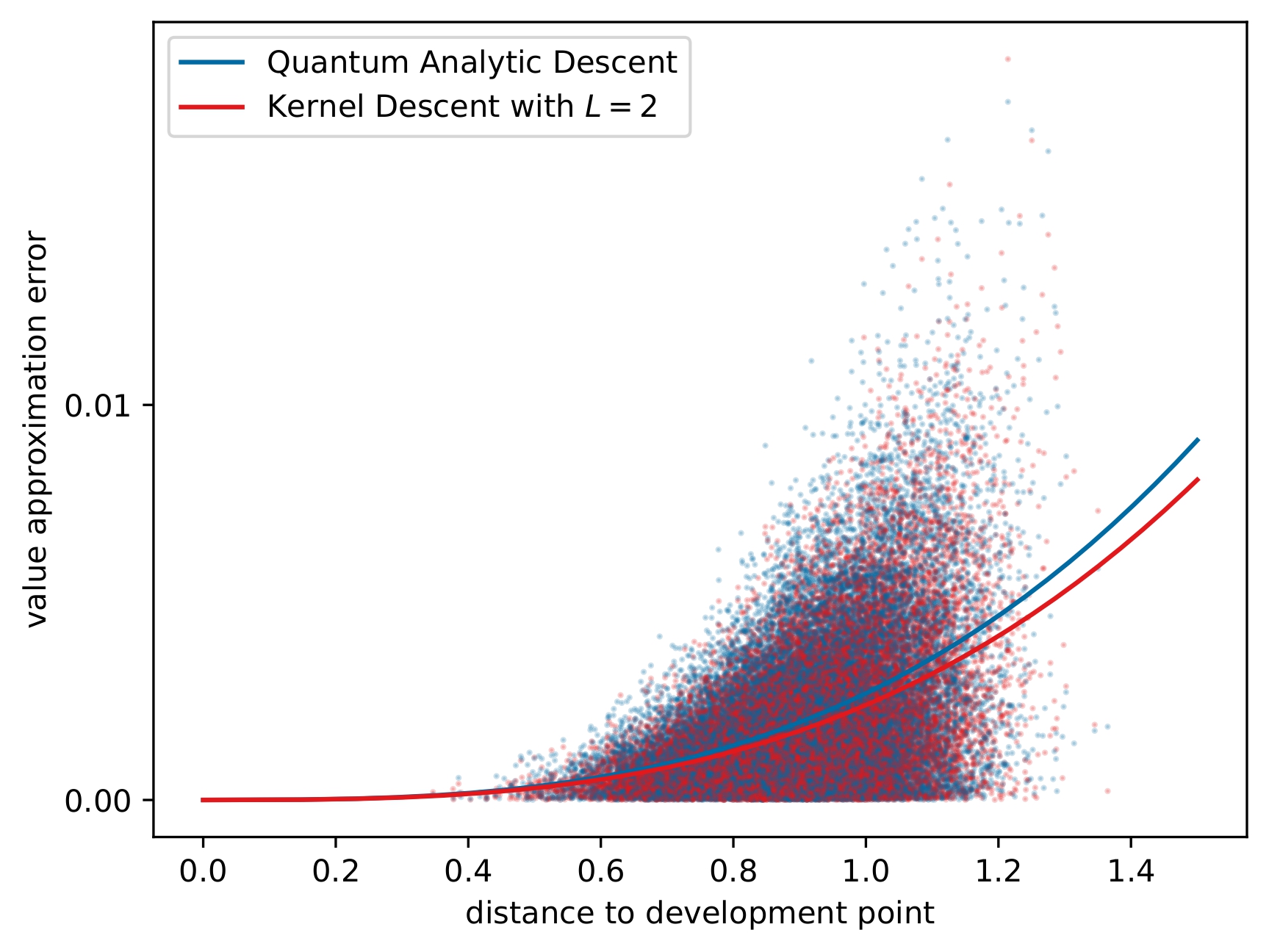}}\hspace{1.0em}%
        \subcaptionbox{Comparison between quantum analytic descent and kernel descent with $L=2$ wrt.\ gradient approximation error in Euclidean norm. We fit curves of the form $x\mapsto c\cdot x^2$ to both point clouds.}{\includegraphics[width=2.13in]{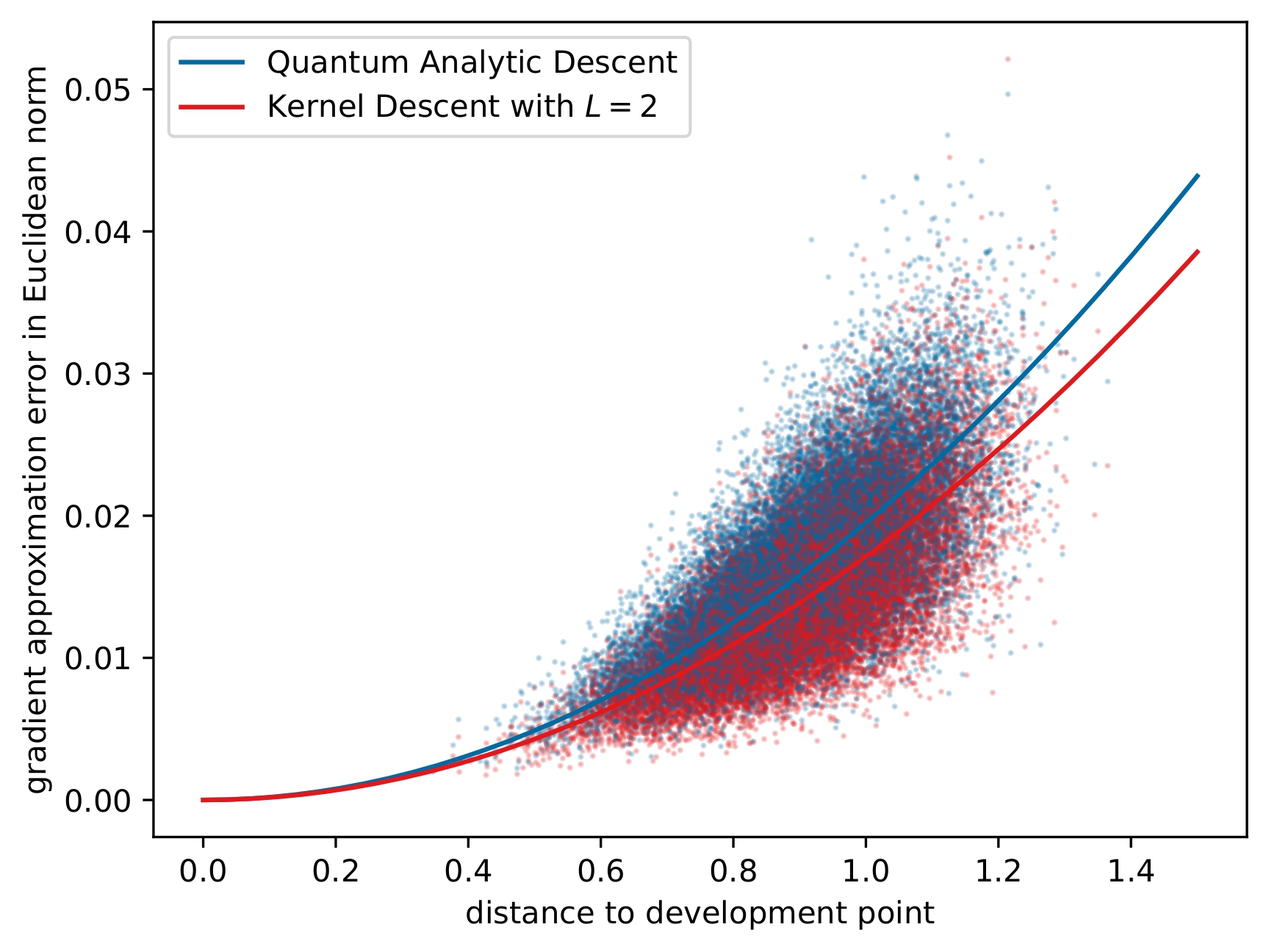}}\hspace{1.0em}%
        \subcaptionbox{Comparison between quantum analytic descent and kernel descent with $L=2$ wrt.\ gradient approximation error in cosine distance. We fit curves of the form $x\mapsto c\cdot x^4$ to both point clouds.}{\includegraphics[width=2.13in]{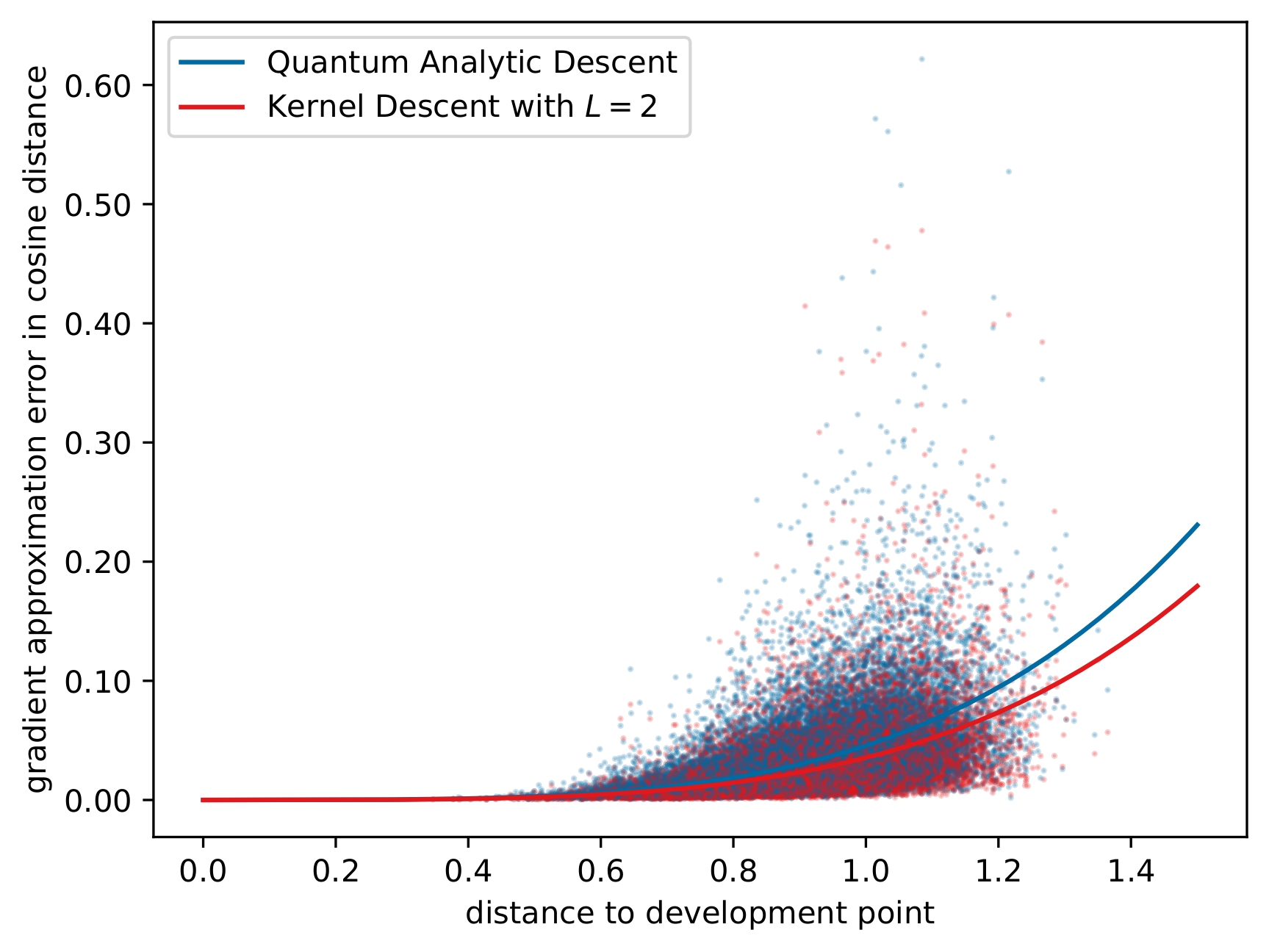}}
     
    	\caption{This figure shows the scatter plots of the second type from Section \ref{subsec:experiments_approx_quality}. The red curve running below the blue curve indicates that kernel descent outperformed the algorithm it was being compared to; see Section \ref{subsec:experiments_approx_quality} for details.}
    	\label{fig:scatter_typeII}
    \end{figure*}

    \subsection{Performance of the algorithms}
    \label{susbsec:experiments_performance}
    Here we compare the algorithms with respect to their ability to minimize the objective function.
    The comparisons of kernel descent with $L=1$ versus gradient descent and kernel descent with $L=2$ versus quantum analytic descent follow a similar pattern, albeit with some key differences which are imposed by the higher computational demands per iteration of the second pair of algorithms.

    \subsubsection{Kernel descent versus gradient descent}
    \label{subsubsec:experiments_kd_vs_gd}  
    To level the playing field, we carry out kernel descent with $L=1$ and gradient descent in their most basic form without specialized enhancements such as learning rate adaptation. In particular, we assume that gradient descent operates with a fixed learning rate $\alpha >0$. Moreover, we rescale the steps of the inner optimization loop of kernel descent so that their combined length equals that of the corresponding gradient descent step. We will now describe the experiments in detail and then conclude with our observations.
    
    \paragraph{Circuit sampling and running the algorithms.}

    In order to keep the computational load manageable, we set $n=m=8$ and choose $T=20$ iterations per algorithm. 
    We will compare the algorithms using $N=5,000$ randomly sampled circuits, and for each circuit we will test three different learning rates $\alpha_1 =7.0$, $\alpha_2 =8.5$ and $\alpha_3 =10.0$. 
    Each of the $N=5,000$ test runs will follow the scheme outlined below.

    \begin{enumerate}
        \item \textbf{Circuit sampling:} A circuit, an initial point $\theta_0\in\mathbb{R}^m$ and an observable $\mathcal{M}$ are sampled randomly as explained in Section \ref{sec:sampling}, and $f$ is defined as in Section \ref{subsec:alg_setting}.
        \item \textbf{Gradient descent runs:} Starting from initial point $\theta_0$, gradient descent with $T=20$ steps is executed three times with the three different learning rates: $\alpha_1$,$\alpha_2$ and $\alpha_3$. This yields three sequences of points in parameter space which we denote by $$\Theta^{\operatorname{GD}}_i\, :=\,\Theta^{(\operatorname{GD}, \alpha_i)}\, :=\, \left(\theta_{0}, \theta_{1}^{(\operatorname{GD}, \alpha_i)},  \dots , \theta_{20}^{(\operatorname{GD}, \alpha_i)}\right)\,,$$ where $i=1,2,3$.
        \item \textbf{Kernel descent runs:} Starting from initial point $\theta_0$, kernel descent with $L=1$ and $T=20$ steps is executed three times:
        \begin{itemize}
            \item Denote the current point in parameter space after $t$ iterations of kernel descent by $\theta_t$, and the local approximation of $f$ around $\theta_t$ by $\tilde{f}_t$. In the inner optimization loop we execute $k=100$ gradient descent steps with respect to $\tilde{f}_t$.
            \item In the three executions of kernel descent, the learning rates for the inner optimization loops are $\alpha_1 / k = 0.07$, $\alpha_2 / k = 0.085$ and $\alpha_3 / k = 0.1$. This rescaling of the learning rates accounts for the number of gradient descent steps taken in the inner loop.
            \item  At each of the $k$ gradient descent steps of the inner optimization loop starting at $\theta_t$, the gradient of the local approximation $\tilde{f}_t$ is rescaled to have length $\Vert\nabla f (\theta_t)\Vert$ (numerical instability and division by $0$ are prevented by adding a small positive constant to the occurring denominators). Since $\nabla f (\theta_t) = \nabla \tilde{f}_t (\theta_t)$ by Theorem \ref{theorem:good_approx}(\ref{theorem:good_approx_coincides_up_to_order_L}), computation of $\nabla f (\theta_t)$ does not require any additional circuit evaluations.
        \end{itemize}
        We thus obtain three sequences of points in parameter space, one per learning rate, which we denote by $$\Theta^{\operatorname{KD}}_i\, :=\,\Theta^{(\operatorname{KD}, \alpha_i)}\, :=\,\left(\theta_{0}, \theta_{1}^{(\operatorname{KD}, \alpha_i)},  \dots , \theta_{20}^{(\operatorname{KD}, \alpha_i)}\right)\,,$$ where $i=1,2,3$.
    \end{enumerate}

    \noindent The final outcome of this procedure is a family $$\Theta\,:=\,\left(\Theta^{\operatorname{A}}_i\right)_{\operatorname{A}\in\{\operatorname{GD}, \operatorname{KD}\},i\in\{1,2,3\}}$$ of six point sequences in parameter space ---one per combination of algorithm and learning rate--- for each of the $N=5,000$ randomly sampled circuits.

\paragraph{Normalization.}
Since the circuit, initial point, and observable were sampled randomly in each of the $N=5,000$ test runs, we cannot directly compare the values of the objective function $f$ between different runs.
We remedy this by normalizing the $f$-values of the point sequences we computed. Specifically, for each of the $N=5,000$ families $\Theta$, the following is carried out:
 Let $v$ be the minimal value of $f$ attained at a point in any of the sequences in $\Theta$. That is, 
        $
            v = \min (\{f(\theta_{0}),f(\theta_{t}^{(\operatorname{GD}, \alpha_i)}), f(\theta_{t}^{(\operatorname{KD}, \alpha_i)}) | t\in \{1,\dots , 20\} , i\in \{1,2,3\}\}).
        $
In the pathological case where $v$ is not smaller than $f(\theta_0)$, one would discard the family $\Theta$ and compute another one. However, this did not happen in our experiments.
Now let 
\begin{equation}\label{eq:l}
    \ell\colon\mathbb{R}\to\mathbb{R},\ x\mapsto \tfrac{x-v}{f(\theta_0)-v}
\end{equation}
 be the uniquely determined affine linear map that maps $f(\theta_0)$ to $1$ and $v$ to $0$. Applying $\ell$ component-wise to the sequences of $f$-values corresponding to the point sequences in $\Theta$, we obtain three normalized sequences for gradient descent,
        \begin{align*}
            \Bar{f}_i^{\operatorname{GD}}\,:=\,\left(1, \ell (f(\theta_{1}^{(\operatorname{GD}, \alpha_i)})) , \dots , \ell (f(\theta_{20}^{(\operatorname{GD}, \alpha_i)}))\right),
        \end{align*}
        and three normalized sequences for kernel descent,
        \begin{align*}
            \Bar{f}_i^{\operatorname{KD}}\,:=\,\left(1, \ell (f(\theta_{1}^{(\operatorname{KD}, \alpha_i)})) , \dots , \ell (f(\theta_{20}^{(\operatorname{KD}, \alpha_i)}))\right),
        \end{align*}
        where $i=1,2,3$.

    \paragraph{Illustration and observations.}
   
   For a fixed learning rate $\alpha_i$, where $i\in\{1,2,3\}$, and a fixed algorithm $\operatorname{A}\in\{\operatorname{GD}, \operatorname{KD}\}$ we can obtain an averaged sequence of normalized $f$-values by taking the component-wise average over all $N=5,000$ normalized sequences $\Bar{f}^{\operatorname{A}}_i$. The so-obtained sequence can be interpreted as the average performance of the considered combination of algorithm and learning rate.
   Comparisons of these average performances are visualized in Figure \ref{fig:lossfct_curves_gd_vs_kd}.

    Upon closer inspection we realized that, for the smallest learning rate $\alpha_1$, the direction of the {\emph{true}} gradient vector did typically not change by much between two consecutive gradient descent steps. Because of this, and since we ensured that the length of the path in parameter space traversed during the inner optimization loop of kernel descent equals the length of the corresponding gradient descent step, one cannot expect kernel descent with $L=1$ to significantly outperform gradient descent at learning rate $\alpha_1$ in our experiments. 
    Nevertheless, as the quality of the local approximations computed by kernel descent with $L=1$ is typically higher than the quality of the best linear approximation (see Section \ref{subsec:experiments_approx_quality}), kernel descent with $L=1$ can still be expected to perform slightly better than gradient descent, even at small learning rates. These expectations are confirmed in Figure \ref{fig:lossfct_curves_gd_vs_kd_learning_rate_alpha1}.

    For the largest learning rate, $\alpha_3$, the steps taken by gradient descent were often too large and the algorithm did not converge. Contrary to this, in the inner optimization loop of the corresponding execution of kernel descent with $L=1$, the descent often got caught in a local minimum of the local approximation, which effectively prevented the algorithm from moving away too far from the development point. As a result, kernel descent typically converged, while gradient descent typically did not. This is mirrored in Figure \ref{fig:lossfct_curves_gd_vs_kd_learning_rate_alpha3}.

    For learning rate $\alpha_2$, we observed a mixture of the phenomena which we observed for learning rates $\alpha_1$ and $\alpha_3$.
    Non-surprisingly, the robustness advantages of kernel descent (against taking exceedingly large steps in parameter space) were less pronounced than with learning rate $\alpha_3$, while the advantages in terms of speed of descent (owing to the better approximation quality, see Section \ref{subsec:experiments_approx_quality}) were more pronounced than with learning rate $\alpha_1$.
    The results for learning rate $\alpha_2$ are visualized in Figure \ref{fig:lossfct_curves_gd_vs_kd_learning_rate_alpha2}.

    Out of all six combinations of algorithm and learning rate, kernel descent with $L=1$ and learning rate $\alpha_3$ performed the best in our experiments. Moreover, every combination involving kernel descent performed better than all combinations involving gradient descent, see Figures \ref{fig:lossfct_curves_gd_vs_kd_kdalpha3_vs_gdall}, \ref{fig:lossfct_curves_gd_vs_kd_kdall_vs_gdall}, \ref{fig:lossfct_curves_gd_vs_kd_kdall_only}.

    In conclusion, kernel descent with $L=1$ outperformed gradient descent at small learning rates in our experiments, which we ascribe to the higher quality of the local approximations, see Section \ref{subsec:experiments_approx_quality}. Moreover, kernel descent with $L=1$ was more robust with respect to the choice of learning rate, which we ascribe to the presence of local minima for the local approximations computed during kernel descent with $L=1$. While the local approximations computed during kernel descent necessarily have local (and global) minima, it is not clear whether the robustness effect we observed in our experiments will continue to hold outside of our experimental setup. An investigation of the latter is left for future work.

    \begin{figure*}[htbp]
        \captionsetup{width=1.0\linewidth}
    	\centering
        \subcaptionbox[width=0.01\linewidth]{Comparison at learning rate $\alpha_1 = 7.0$.\label{fig:lossfct_curves_gd_vs_kd_learning_rate_alpha1}}{\includegraphics[width=2.13in]{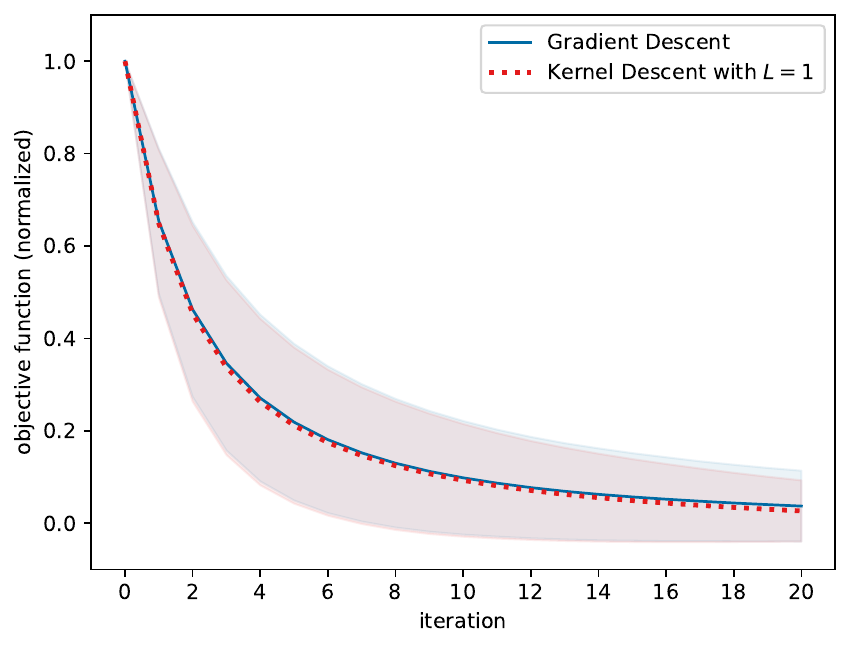}}\hspace{1.0em}%
        \subcaptionbox{Comparison at learning rate $\alpha_2 = 8.5$.\label{fig:lossfct_curves_gd_vs_kd_learning_rate_alpha2}}{\includegraphics[width=2.13in]{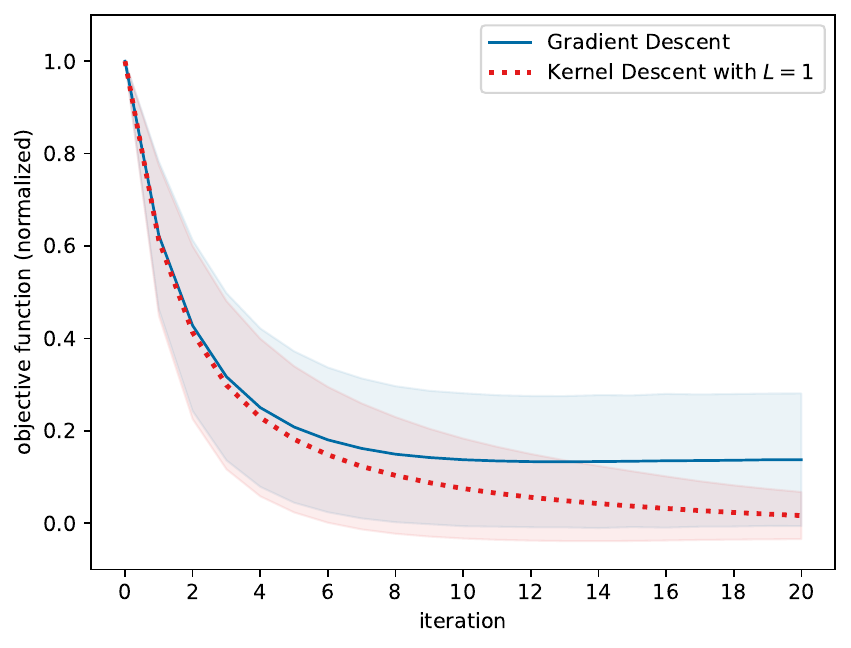}}\hspace{1.0em}%
        \subcaptionbox{Comparison at learning rate $\alpha_3 = 10.0$.\label{fig:lossfct_curves_gd_vs_kd_learning_rate_alpha3}}{\includegraphics[width=2.13in]{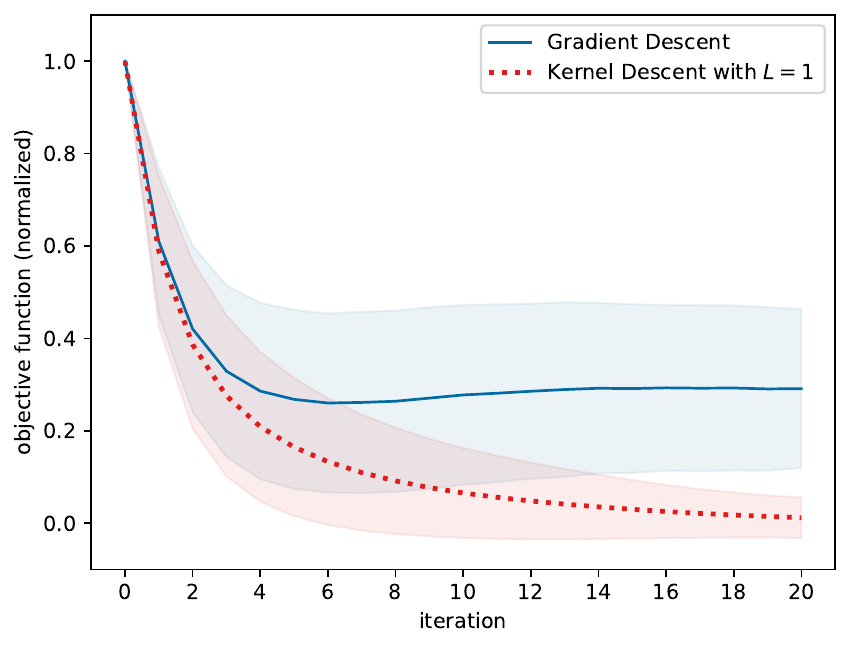}}
    	
    	\bigskip
    	
        \subcaptionbox{Comparison between gradient descent at learning rates $\alpha_1$, $\alpha_2$, $\alpha_3$ and kernel descent with $L=1$ at learning rate $\alpha_3$.\label{fig:lossfct_curves_gd_vs_kd_kdalpha3_vs_gdall}}{\includegraphics[width=2.13in]{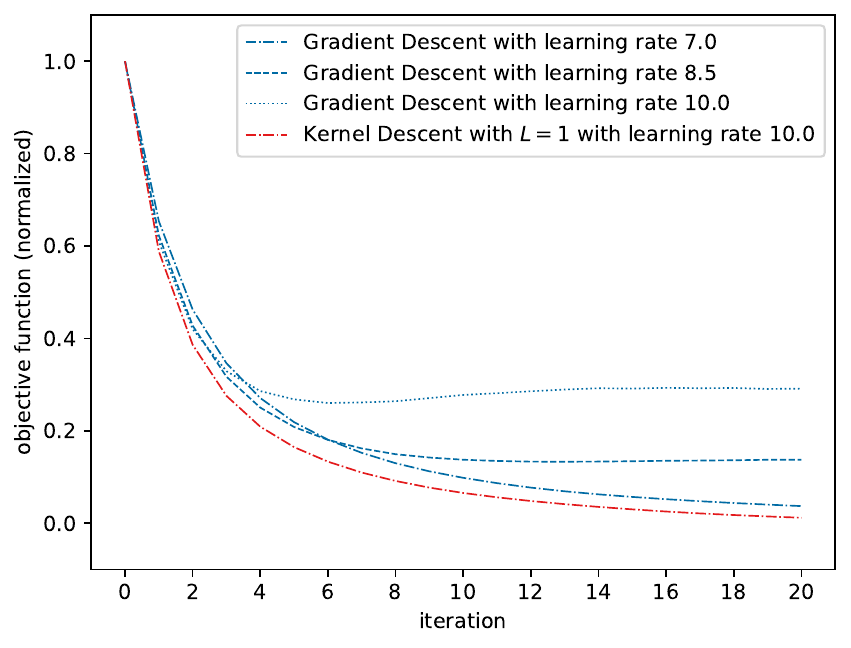}}\hspace{1.0em}%
        \subcaptionbox{Comparison between gradient descent and kernel descent with $L=1$ at learning rates $\alpha_1$, $\alpha_2$, $\alpha_3$.\label{fig:lossfct_curves_gd_vs_kd_kdall_vs_gdall}}{\includegraphics[width=2.13in]{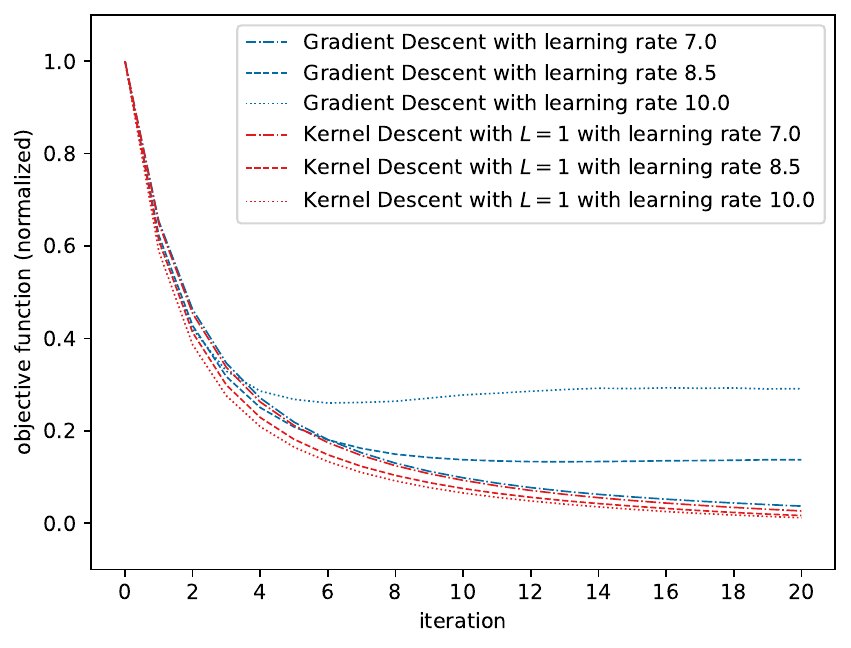}}\hspace{1.0em}%
        \subcaptionbox{Comparison between learning rates $\alpha_1$, $\alpha_2$, $\alpha_3$ for kernel descent with $L=1$.\label{fig:lossfct_curves_gd_vs_kd_kdall_only}}{\includegraphics[width=2.13in]{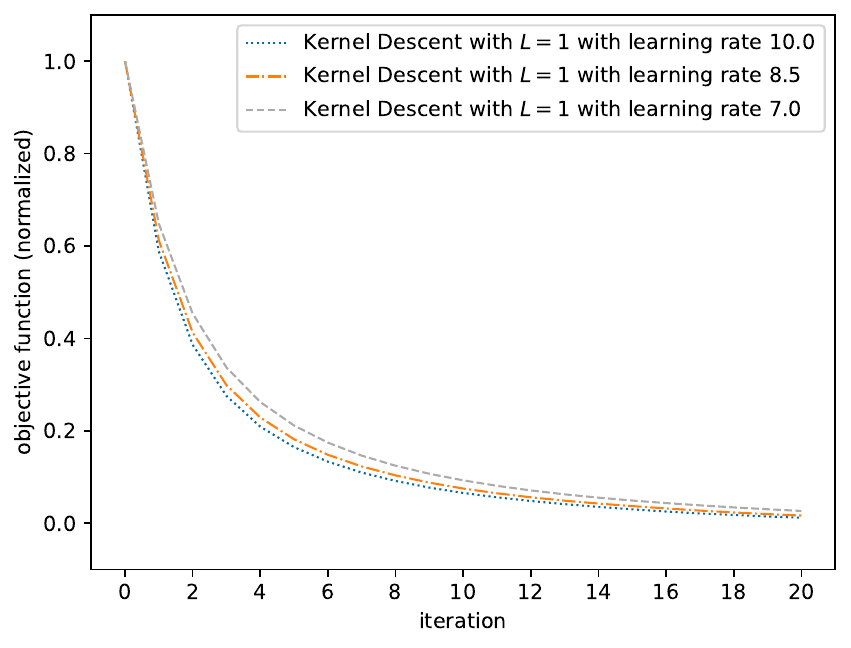}}
     
    	\caption{This figure shows the performance of gradient descent and kernel descent with $L=1$ for three different learning rates, averaged over $N=5,000$ respective executions. In each execution, the circuit, observable and initial point in parameter space were sampled randomly. Consequently, for the sake of comparability, we introduced a normalization step before averaging; see Section \ref{subsubsec:experiments_kd_vs_gd} for details. In Figures \ref{fig:lossfct_curves_gd_vs_kd_learning_rate_alpha1}, \ref{fig:lossfct_curves_gd_vs_kd_learning_rate_alpha2}, \ref{fig:lossfct_curves_gd_vs_kd_learning_rate_alpha3} we compare gradient descent to kernel descent with $L=1$ at learning rates $\alpha_1 =7.0$, $\alpha_2 =8.5$ and $\alpha_3 =10.0$ respectively; we also indicate the respective standard deviations. In Figures \ref{fig:lossfct_curves_gd_vs_kd_kdalpha3_vs_gdall}, \ref{fig:lossfct_curves_gd_vs_kd_kdall_vs_gdall}, \ref{fig:lossfct_curves_gd_vs_kd_kdall_only} we make comparisons between different learning rates.}
    	\label{fig:lossfct_curves_gd_vs_kd}
    \end{figure*}

    \subsubsection{Kernel descent versus quantum analytic descent}
    \label{subsubsec:experiments_kd_vs_qad}
    The comparison of kernel descent with $L=2$ versus quantum analytic descent follows a similar scheme as the comparison of kernel descent with $L=1$ versus gradient descent. Key differences concern the random sampling of the observable and the stopping criteria for the inner optimization loops of the algorithms.

    \paragraph{Circuit sampling.}

    The circuit and the initial point in parameter space are sampled precisely as described in Section \ref{sec:sampling}. The observable $\mathcal{M}$ is sampled as follows: We randomly sample (independently) twenty $n$-qubit Paulis $\mathcal{P}_1 , \dots , \mathcal{P}_{20}$ from the uniform distribution on $\{I,X,Y,Z\}^{\otimes n}$; we also randomly sample (independently) coefficients $c_1 , \dots , c_{20}$ from the standard normal distribution. We then set $\mathcal{M}:=\sum_{i=1}^{20}c_i\mathcal{P}_i$.

    \paragraph{The inner optimization loop.} 
    To ensure a fair comparison, the inner optimization loop was conducted in the same way for kernel descent and quantum analytic descent; specifically, we used gradient descent with learning rate $0.01$ and applied the same stopping criterion to both algorithms.
    In order to be as impartial as possible, we used one of the stopping criteria given in the original work on quantum analytic descent \cite{quantum_analytic_descent}:

    The inner optimization loop is stopped as soon as the {\emph{true}} value of $f$ is increased. Since every determination of the true value of $f$ requires an additional circuit evaluation, we only determine the true value of $f$ every $1000^\text{th}$ step of the inner loop. Moreover, in order to prevent an infinite loop and to upper-bound the number of circuit evaluations per iteration of the algorithms, we set an upper bound of $10000$ steps after which the inner optimization loop is aborted, even if the true value of $f$ has not been found to increase. Consequently, at most $9$ additional circuit evaluations are incurred as a result of the inner optimization loop -- which is negligible in comparison to the overall number of circuit evaluations per iteration. 
    
    \paragraph{Experiment execution.} 
    Due to the aforementioned higher computational cost per iteration the number of test runs for the current comparison is reduced from $N=5,000$ to $N=500$.
    In each run the following is performed:
    
    \begin{enumerate}
        \item A circuit, an initial point $\theta_0\in\mathbb{R}^m$ and an observable $\mathcal{M}$ are sampled randomly as explained previously.
        \item Starting from initial point $\theta_0$, both quantum analytic descent and kernel descent with $L=2$ are executed each with $T=5$ steps (with inner optimization loop as described above).
        We denote the resulting sequences of points in parameter space by
        $$
        \Theta^{\operatorname{QAD}}\,:=\,\left(\theta_{0}, \theta_{1}^{(\operatorname{QAD)}},  \dots , \theta_{5}^{(\operatorname{QAD})}\right)
        $$
        and
        $$
        \Theta^{\operatorname{KD}}\,:=\,\left(\theta_{0}, \theta_{1}^{(\operatorname{KD})},  \dots , \theta_{5}^{(\operatorname{KD})}\right)\,,
        $$
        respectively.
    \end{enumerate}
\paragraph{Normalization.}
As in Section \ref{subsubsec:experiments_kd_vs_gd} we need to introduce a normalization step before we can compare the values of the objective function $f$ between different runs. For each of the $N=500$ runs we proceed as follows: We denote the minimal value of $f$ attained during execution of the two algorithms by $v$. That is, 
$$
            v = \min (\{f(\theta_0 ), f(\theta_{t}^{(\operatorname{QAD})}), f(\theta_{t}^{(\operatorname{KD})}) | t\in \{1,\dots , 5\}\}).
$$
As above, no pathological cases (where $v=f(\theta_0)$) arose and thus no exception handling was necessary.

Using $v$, we define $\ell: \mathbb{R}\to\mathbb{R}$ as in \eqref{eq:l}. Applying $\ell$ component-wise to the sequences of values obtained during execution of the algorithms, we obtain normalized sequences 
\begin{align*}
            \Bar{f}^{\operatorname{QAD}}\,:=\,\left(1, \ell (f(\theta_{1}^{(\operatorname{QAD})})) , \dots , \ell (f(\theta_{5}^{(\operatorname{QAD})}))\right),
        \end{align*}
        and
        \begin{align*}
            \Bar{f}^{\operatorname{KD}}\,:=\,\left(1, \ell (f(\theta_{1}^{(\operatorname{KD})})) , \dots , \ell (f(\theta_{5}^{(\operatorname{KD})}))\right)\,.
        \end{align*}
As in Section \ref{subsubsec:experiments_kd_vs_gd}, we then compute the component-wise average of the $N=500$ normalized sequences $\Bar{f}^{\operatorname{QAD}}$ (resp.\ $\Bar{f}^{\operatorname{KD}}$).

    \paragraph{Results.} Kernel descent with $L=2$ clearly outperformed quantum analytic descent, in the sense that the average curve for kernel descent runs below that of quantum analytic descent, see Figure \ref{fig:lossfct_qad_vs_kd}.

    \begin{FigureInColumn}
        \centering
            \includegraphics[width = 3.0in]{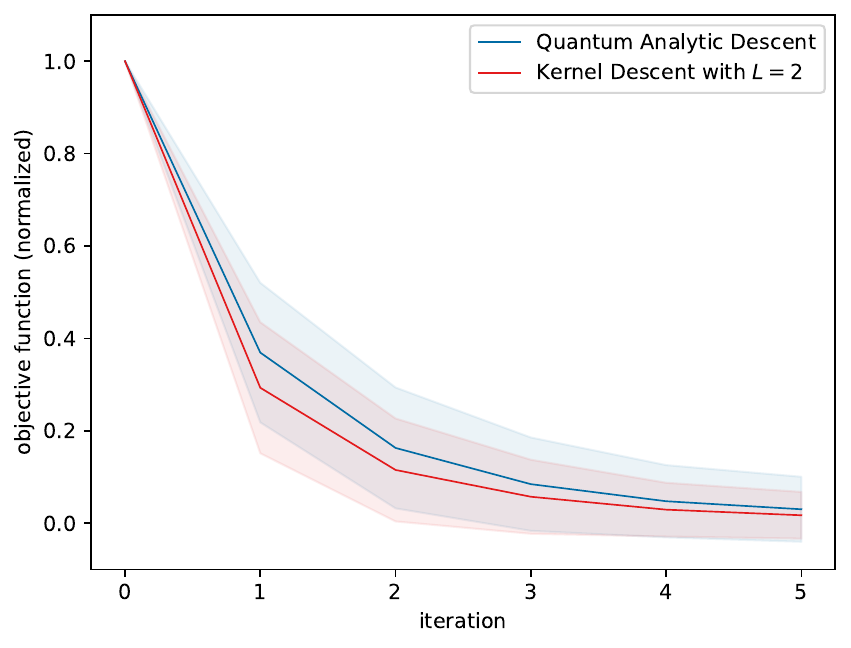}
        \captionof{figure}{This figure shows the performance of quantum analytic descent and kernel descent with $L=2$ averaged over $N = 500$ respective executions. In each execution, the circuit, observable and initial point in parameter space were sampled randomly. Consequently, for the sake of comparability, we introduced a normalization step before averaging; see Section \ref{subsubsec:experiments_kd_vs_qad} for details. We also indicate the respective standard deviations.}
        \label{fig:lossfct_qad_vs_kd}
    \end{FigureInColumn}
    
\subsection{Experiment with $8$-qubit spin ring Hamiltonian}\label{subsec:experiments_spin_ring_hamiltonian}

The first part of Section \ref{sec:experiments}, described in Sections \ref{sec:sampling}, \ref{subsec:experiments_approx_quality}, and \ref{susbsec:experiments_performance}, was concerned with a large number of randomized circuits and observables. Here, in the second part, we experimentally compare the performance of gradient descent and kernel descent with $L=1$ for an observable and a circuit that are more representative of real world applications.

\paragraph{The problem.} We address the problem of finding the ground state energy of an $8$-qubit spin ring Hamiltonian, which was already investigated in \cite{quantum_analytic_descent}. Specifically, we consider the Hamiltonian
$$\sum_{i=1}^{8} 0.1\cdot(X_i X_{i+1} + Y_i Y_{i+1} + Z_i Z_{i+1}) + \sum_{i=1}^{8} \omega_i Z_i,$$
where all indices are taken modulo $8$ with system of representatives $\{1,2,\dots , 8\}$, and $\omega_1,\dots , \omega_8\in [-1,1]$ were independently sampled uniformly at random.

\paragraph{Ansatz circuit.} We adopt the ansatz circuit from \cite{quantum_analytic_descent}, i.e., starting with $n=8$ qubits in state $|0\rangle^{\otimes 8}$, we apply a parametrized $R_X$ gate to every qubit and subsequently apply four consecutive blocks, each consisting of
\begin{enumerate}
    \item parametrized $R_{ZZ}$ gates, applied to qubit pairs $(7, 8)$, $(6, 7)$, $(5, 6)$, $(4, 5)$, $(3, 4)$, $(2, 3)$, $(1, 2)$, and $(8, 1)$,
    \item followed by parametrized $R_Y$ gates applied to every qubit,
    \item followed by parametrized $R_X$ gates applied to every qubit.
\end{enumerate}
In particular, we have $m=8+4\cdot(8+8+8) = 104$. The circuit is visualized in Figure \ref{fig:circuit_revision_1}.

\tikzset{
    operator/.append style={draw,fill=blue!10, text height = 50pt, text width = 20pt},
}

\begin{figure*}[htbp]
    \centering
        \newsavebox{\quantikzbox}
 \sbox{\quantikzbox}{%
\begin{quantikz}[row sep=0.3cm]  
\lstick{\( \ket{0} \)} & \gate[style={text height = 20pt} ]{R_X} & \qw & \qw & \qw & \qw & \qw & \qw & \gate[wires=2]{R_{ZZ}} & \gate[wires=8]{\makebox[3pt]{\raisebox{4pt}[13pt][0pt]{$\begin{array}{c} \vspace{-.22cm} {\scriptscriptstyle 2}  \\ \\  \\ \\ \\ \\ \\  \\ \vspace{+.2cm} R_{ZZ}  \\ \\ \\ \\ \\ \\ \\ {\scriptscriptstyle 1} \end{array}$}}} & \gate[style={text height = 20pt}]{R_Y} & \gate[style={text height = 20pt}]{R_X} & \qw \\
\lstick{\( \ket{0} \)} & \gate[style={text height = 20pt} ]{R_X} & \qw & \qw & \qw & \qw & \qw & \gate[wires=2]{R_{ZZ}} & \qw & \qw & \gate[style={text height = 20pt} ]{R_Y} & \gate[style={text height = 20pt} ]{R_X} & \qw \\
\lstick{\( \ket{0} \)} & \gate[style={text height = 20pt} ]{R_X} & \qw & \qw & \qw & \qw & \gate[wires=2]{R_{ZZ}} & \qw & \qw & \qw & \gate[style={text height = 20pt} ]{R_Y} & \gate[style={text height = 20pt} ]{R_X} & \qw \\
\lstick{\( \ket{0} \)} & \gate[style={text height = 20pt} ]{R_X} & \qw & \qw & \qw & \gate[wires=2]{R_{ZZ}} & \qw & \qw & \qw & \qw & \gate[style={text height = 20pt} ]{R_Y} & \gate[style={text height = 20pt} ]{R_X} & \qw \\
\lstick{\( \ket{0} \)} & \gate[style={text height = 20pt} ]{R_X} & \qw & \qw & \gate[wires=2]{R_{ZZ}} & \qw & \qw & \qw & \qw & \qw & \gate[style={text height = 20pt} ]{R_Y} & \gate[style={text height = 20pt} ]{R_X} & \qw \\
\lstick{\( \ket{0} \)} & \gate[style={text height = 20pt} ]{R_X} & \qw & \gate[wires=2]{R_{ZZ}} & \qw & \qw & \qw & \qw & \qw & \qw & \gate[style={text height = 20pt} ]{R_Y} & \gate[style={text height = 20pt} ]{R_X} & \qw \\
\lstick{\( \ket{0} \)} & \gate[style={text height = 20pt} ]{R_X} & \gate[wires=2]{R_{ZZ}} & \qw & \qw & \qw & \qw & \qw & \qw & \qw & \gate[style={text height = 20pt} ]{R_Y} & \gate[style={text height = 20pt} ]{R_X} & \qw \\
\lstick{\( \ket{0} \)} & \gate[style={text height = 20pt} ]{R_X} & \qw & \qw & \qw & \qw & \qw & \qw & \qw & \qw & \gate[style={text height = 20pt} ]{R_Y} & \gate[style={text height = 20pt} ]{R_X} & \qw \\
\end{quantikz}
}
$\Ubrace[60pt]{\hspace{-48pt}\usebox{\quantikzbox}}{U}[\dots]$
    \caption{This figure shows the ansatz circuit used for the experiments in Section \ref{subsec:experiments_spin_ring_hamiltonian}. While, for clarity, the $U$-block is depicted only once, the circuit consists of an initial layer of parametrized $R_X$ gates and four consecutive $U$-blocks (with different parameters, i.e., we are not re-uploading parameters). In particular, the ansatz circuit has $m=8+4\cdot (8+8+8)=104$ parameters.}\label{fig:circuit_revision_1}
\end{figure*}

\paragraph{Experiment execution.} All circuit evaluations were executed using $1000$ measurement shots. After sampling an initial point in parameter space from the uniform distribution on $[-\pi ,\pi )^{104}$, both gradient descent and kernel descent with $L=1$ were executed $N=100$ times, with $30$ iterations per execution, yielding $N=100$ respective sequences in parameter space of length $31$ (including the initial point). For each such sequence in parameter space we then determined the corresponding sequence of true values of the objective function $f$ using statevector simulation (exact evaluations of $f$  did not have an impact on the execution of either algorithm, but were only used to compare the results). Subsequently, for both algorithms, we computed the component-wise average of the $N=100$ sequences of values to get an indication of their respective average performance. However, since the stopping criterion for kernel descent we used in this experiment incurs additional circuit evaluations in the inner optimization loop (see below), we compare the performance of the algorithms based on the number of executed measurement shots. Since, when using the above-mentioned stopping criterion for the inner optimization loop, the number of measurement shots per iteration of kernel descent is not constant, we also averaged the cumulative number of measurement shots for each of the $30$ iterations over the $N=100$ executions of kernel descent.

\paragraph{Gradient descent.} The learning rate for gradient descent was $0.3$ (the largest learning rate under which stable behavior was observed), and gradients were evaluated using the parameter-shift rules with $\pm\frac{\pi}{2}$-shifts.

\paragraph{Kernel descent with $L=1$.} We employed a similar stopping criterion for the inner optimization loop as in Section \ref{subsubsec:experiments_kd_vs_qad}. Specifically, we used gradient descent with learning rate $0.001$ in the inner optimization loop and terminated the inner optimization loop as soon as the value of $f$ was increased; the value of $f$ (up to shot noise) was determined every $100^\text{th}$ step of the inner loop and we set an upper bound of $5000$ steps for the inner loop, after which the loop was aborted even if the value of $f$ was not found to increase. The upside of using this stopping criterion is that one can forego a (potentially lengthy) tuning process for the hyperparameters of the inner optimization loop. The downside is that this stopping criterion incurs several additional circuit evaluations in each iteration of the outer optimization loop. In order to not give kernel descent an unfair advantage over gradient descent (the additional circuit evaluations), we compare the performance of the algorithms based on the number of executed measurement shots (and not based on the number of executed iterations), see above.

\paragraph{Results.} Kernel descent with $L = 1$ clearly outperformed gradient descent, in the sense that the average curve for kernel descent runs below that of gradient descent, see Figure \ref{fig:lossfct_spin_ring_hamiltonian}.

    \begin{FigureInColumn}
        \centering
            \includegraphics[width = 3.0in]{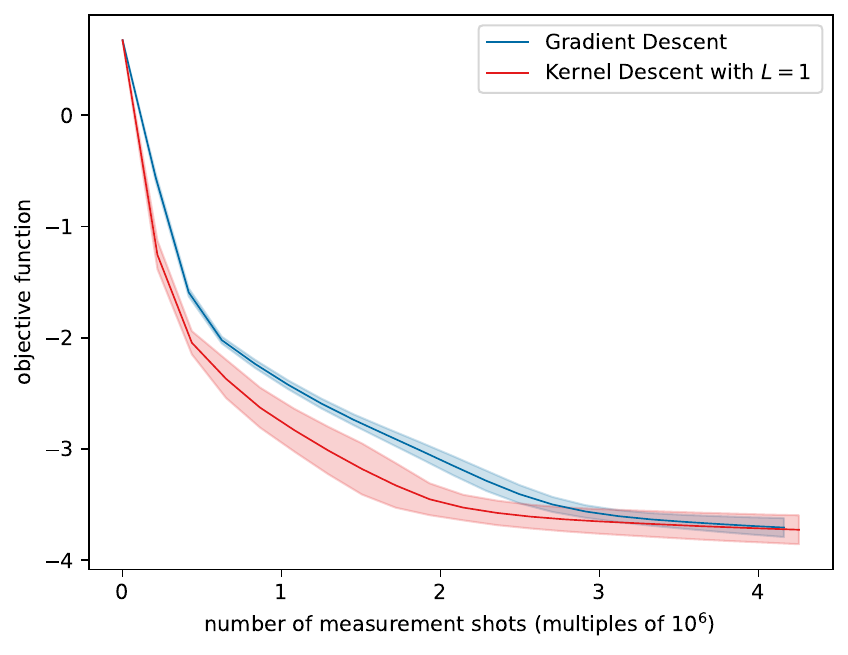}
        \captionof{figure}{This figure shows the performance (under the presence of measurement shot noise) of gradient descent and kernel descent with $L=1$ in finding the ground state energy of an $8$-qubit spin ring Hamiltonian, averaged over $100$ respective executions. Since the stopping criterion we used for the inner optimization loop of kernel descent incurs additional circuit evaluations, we compare the performance of the algorithms based on the number of executed measurement shots; see Section \ref{subsec:experiments_spin_ring_hamiltonian} for details. We also indicate the respective standard deviations.}
        \label{fig:lossfct_spin_ring_hamiltonian}
    \end{FigureInColumn}

    \section{CONCLUSION}
    \label{sec:conclusion}
    
    In this article we introduced \algo\ (\algoAbb), a novel algorithm for minimizing the functions underlying variational quantum algorithms. More precisely, \algoAbb\ is an algorithm for finding a choice of parameters for a parametrized quantum circuit that minimizes the expected value of a given observable with respect to the state computed by this circuit.
    
    As is the case for quantum analytic descent, each iteration of our algorithm involves using a quantum device to construct a classical local approximation to the objective function and subsequently carrying out a classical optimization routine in order to minimize this approximation. The local approximations are constructed by exploiting the fact that the functions computed by variational quantum algorithms are contained in a set which naturally carries the structure of a reproducing kernel Hilbert space.
    
    This added structure also serves as a heuristic explanation for the fact that the local approximations constructed during execution of our algorithm significantly outperform $L$-th order Taylor approximation in our experiments despite not requiring a larger number of quantum circuit evaluations to compute. Here, $L$ is a hyperparameter which controls the order of the local approximations featuring in our algorithm. After theoretically comparing our algorithm to gradient descent and to quantum analytic descent, we carried out extensive experiments to demonstrate the validity of our approach.

    Potential topics for future research include, but are not limited to, the following:
    \begin{itemize}
        \item In the present work we considered parametrized gates of the form $x\mapsto \exp (-i \tfrac{x}{2}G)$ whose Hermitian generators $G$ were assumed to have as set of eigenvalues $\{-1,1\}$. While this includes the often-studied case where $G$ is a (non-identity) tensor product of Pauli matrices, this still limits the generality of our approach. In light of this, a natural direction for future work would be to extend our results by imposing less restrictive assumptions on the eigenvalues of the generators of the occurring parametrized gates.
        \item It would be useful to develop an (adaptive) measurement shot allocation strategy as in \cite{quantum_analytic_descent}, \cite{ito2023} and \cite{tamiya2022}.
        \item Building on our experiments in Section \ref{sec:experiments}, which take measurement shot noise into account, a logical next step would be to conduct experiments under the presence of quantum hardware noise -- either real or simulated.
        \item As stated in Section \ref{sec:introduction}, advanced techniques, such as normalized gradient descent \cite{Hazan2015BeyondCS}, Nesterov’s Accelerated Gradient method \cite{Nesterov1983AMF}, the ADAM optimizer \cite{KingBa15}, or natural gradient descent \cite{natural_gradient_descent}, can be adapted to kernel descent. Describing these enhancements in detail and exploring them experimentally would be a natural continuation of this work.
    \end{itemize}


    \section*{DATA AVAILABILITY}
The datasets used and/or analysed during the current study are available from the corresponding author on reasonable request.
    
    \section*{ACKNOWLEDGEMENT}
    This article was written as part of the Qu-Gov project, which was commissioned by the German Federal Ministry of Finance. The authors want to extend their gratitude to Kim Nguyen, Manfred Paeschke, Oliver Muth, Andreas Wilke, and Yvonne Ripke for their continuous encouragement and support.

    \bibliographystyle{amsplain}
	\bibliography{literature/bibliography}

\end{multicols}

\newpage
\begin{appendices}
    Here we give a proof for Theorem \ref{theorem:good_approx}(\ref{theorem:good_approx_error_estimate}). In fact, we prove something slightly more general:

    \begin{lemma}
        \label{lemma:error_estimate_when_vanishing_along_certain_subspaces}
        Let $L,m\in\mathbb{Z}_{\geq 1}$ be positive integers, let $\Vert\cdot\Vert$ be a norm on $\mathbb{R}^m$, let $f\colon\mathbb{R}^m\to\mathbb{R}$ be a smooth function, and assume that
        \begin{itemize}
            \item $(D^{\alpha}f)(0) = 0$ for all multiindices $\alpha\in (\mathbb{Z}_{\geq 0})^m$ with $|\alpha |\leq L$,
            \item $f=0$ on some closed set $\mathcal{C}\subseteq\mathbb{R}^m$ satisfying $0\in\mathcal{C}$.
        \end{itemize}
        Then there exist an open neighborhood $U\subseteq\mathbb{R}^m$ of $0$ and a constant $C>0$, such that the estimate 
        \begin{align*}
            |f(x)|
            \leq
            C\Vert x \Vert^L \operatorname{dist}_{\Vert\cdot\Vert}(x, \mathcal{C})
        \end{align*}
        holds for all $x\in U$.
    \end{lemma}

    \begin{proof}
        Since all norms on $\mathbb{R}^m$ are pairwise equivalent, we can assume without loss of generality that $\Vert\cdot\Vert$ is the $2$-norm on $\mathbb{R}^m$. By assumption, for all $j\in\{1,\dots , m\}$, the function $\frac{\partial f}{\partial x_j}\colon\mathbb{R}^m\to\mathbb{R}$ is smooth and satisfies $\left(D^{\beta}\frac{\partial f}{\partial x_j}\right)(0) = 0$ for all multiindices $\beta\in (\mathbb{Z}_{\geq 0})^m$ with $|\beta |\leq L-1$. In particular, for all $j$, we find a constant $C_j >0$ and an open neighborhood $U_j\subseteq\mathbb{R}^m$ of $0$, such that the inequality
        \begin{align*}
            \left|\frac{\partial f}{\partial x_j}(x)\right|
            & = \left|\frac{\partial f}{\partial x_j}(x) - \sum_{\beta\in (\mathbb{Z}_{\geq 0})^m\colon |\beta|\leq L-1} \frac{\left(D^{\beta}\frac{\partial f}{\partial x_j}\right)(0)}{\beta!}x^\beta\right|\\
            & \leq
            C_j \Vert x\Vert^L
        \end{align*}
        holds for all $x\in U_j$. Now pick some $\delta >0$ with $B_\delta (0)\subseteq U_1\cap\dots\cap U_m$ (here, $B_\delta (0)$ denotes the open ball of radius $\delta$ centered at $0$ in $\mathbb{R}^m$). Since $\mathcal{C}$ contains $0$ and is thus non-empty, the map $\operatorname{dist}_{\Vert\cdot\Vert}(\cdot , \mathcal{C})\colon\mathbb{R}^m\to\mathbb{R}$, $x\mapsto \operatorname{dist}_{\Vert\cdot\Vert}(x , \mathcal{C}) = \inf_{z\in\mathcal{C}}\Vert x-z\Vert$ is well-defined and continuous. Hence, due to compactness, this map attains its maximum on the closed ball $\overline{B_\delta (0)}$. In particular, there exists some $R>0$, such that $\operatorname{dist}_{\Vert\cdot\Vert}(x , \mathcal{C}) < R$ for all $x\in\overline{B_\delta (0)}$. Moreover, for all $j\in\{1,\dots , m\}$, the mapping
        \begin{align*}
            \overline{B_{\delta + R} (0)}\setminus B_\delta (0) \to \mathbb{R}_{\geq 0}, x\mapsto \frac{\left|\frac{\partial f}{\partial x_j}(x)\right|}{\Vert x\Vert^L}
        \end{align*}
        is well-defined, continuous, and hence -- due to compactness -- bounded from above by some constant $\tilde{C}_j >0$. Define $\tilde{C}:=\max\{C_1 , \dots , C_m , \tilde{C}_1 , \dots , \tilde{C}_m\}$. We then have for all $j\in\{1,\dots , m\}$ and $x\in \overline{B_{\delta + R} (0)}$:
        \begin{align*}
            \left|\frac{\partial f}{\partial x_j}(x)\right|
            & = 
            \begin{cases}
                \left|\frac{\partial f}{\partial x_j}(x)\right| & \text{if }x\in B_{\delta} (0)\\
                \frac{\left|\frac{\partial f}{\partial x_j}(x)\right|}{\Vert x\Vert^L}\cdot\Vert x\Vert^L  & \text{if }x\in \overline{B_{\delta + R} (0)}\setminus B_\delta (0)
            \end{cases}\\
            & \leq
            \begin{cases}
                C_j\Vert x\Vert^L & \text{if }x\in B_{\delta} (0)\\
                \tilde{C}_j\Vert x\Vert^L  & \text{if }x\in \overline{B_{\delta + R} (0)}\setminus B_\delta (0)
            \end{cases}\\
            & \leq
            \tilde{C} \Vert x\Vert^L .
        \end{align*}
        Now pick and fix an arbitrary point $x\in B_\delta (0)$. Since $\emptyset\neq\mathcal{C}$ is closed in $\mathbb{R}^m$, there exists a $z\in\mathcal{C}$, such that $\Vert x-z\Vert = \operatorname{dist}_{\Vert\cdot\Vert}(x , \mathcal{C})$. Since $\operatorname{dist}_{\Vert\cdot\Vert}(x , \mathcal{C}) < R$ (owing to the fact that $x\in  B_\delta (0)\subseteq\overline{ B_\delta (0)}$), we have
        $$
        \Vert z\Vert =\Vert x + (z-x)\Vert \leq \Vert x\Vert + \Vert z-x\Vert = \Vert x\Vert +\operatorname{dist}_{\Vert\cdot\Vert}(x , \mathcal{C}) < \delta + R,
        $$
        hence $z\in B_{\delta + R} (0)$. By convexity we then have $z + t(x-z)\in B_{\delta + R} (0)$ for all $t\in[0,1]$. Moreover, since $0\in\mathcal{C}$, we have
        \begin{align*}
        \Vert z + t(x-z) \Vert
        & =
        \Vert x+(1-t)(z-x) \Vert\\
        & \leq
        \Vert x\Vert + (1-t)\Vert z-x \Vert\\
        & =
        \Vert x\Vert + (1-t)\operatorname{dist}_{\Vert\cdot\Vert}(x , \mathcal{C})\\
        & \leq
        \Vert x\Vert + (1-t) \Vert x-0\Vert\\
        & \leq
        2\Vert x\Vert
        \end{align*}
        for all $t\in[0,1]$. Since $\Vert\cdot\Vert$ is the $2$-norm on $\mathbb{R}^m$, it obeys the Cauchy-Schwarz inequality with the canonical inner product on $\mathbb{R}^m$. Letting $\gamma\colon\mathbb{R}\to\mathbb{R}^m$, $t\mapsto z + t(x-z)$, we then calculate, combining all of the above and using that $f(z)=0$, since $f$ vanishes on $\mathcal{C}$:
        \begin{align*}
            |f(x)|
            & =
            |f(x)-f(z)|\\
            & =
            |f(\gamma (1)) - f(\gamma (0))|\\
            & =
            \left|\int_0^1 (f\circ\gamma)'(t) dt\right|\\
            & \leq
            \int_0^1 \left|(f\circ\gamma)'(t)\right| dt\\
            & =
            \int_0^1 \left|\underbrace{(\nabla f) (\gamma(t))}_{\in\mathbb{R}^{1\times m}} \cdot \underbrace{\gamma'(t)}_{\in\mathbb{R}^{m\times 1}}\right| dt\\
            & \leq
            \int_0^1 \Vert(\nabla f) (\gamma(t))\Vert \cdot \Vert\gamma'(t)\Vert dt && \text{(by Cauchy-Schwarz)}\\
            & =
            \int_0^1 \sqrt{\sum_{j=1}^m \left(\frac{\partial f}{\partial x_j}(z + t(x-z))\right)^2}\cdot\Vert x-z\Vert dt\\
            & = 
            \operatorname{dist}_{\Vert\cdot\Vert}(x , \mathcal{C})\cdot\int_0^1 \sqrt{\sum_{j=1}^m \left(\frac{\partial f}{\partial x_j}(\underbrace{z + t(x-z)}_{\in B_{\delta + R}(0)})\right)^2} dt\\
            & \leq
            \operatorname{dist}_{\Vert\cdot\Vert}(x , \mathcal{C})\cdot\int_0^1 \sqrt{\sum_{j=1}^m \left(\tilde{C}\Vert z + t(x-z)\Vert^L\right)^2} dt\\
            & \leq
            \operatorname{dist}_{\Vert\cdot\Vert}(x , \mathcal{C})\cdot\int_0^1 \sqrt{\sum_{j=1}^m \left(2^L\tilde{C}\Vert x\Vert^L\right)^2} dt\\
            & =
            \operatorname{dist}_{\Vert\cdot\Vert}(x , \mathcal{C})\cdot \sqrt{\sum_{j=1}^m \left(2^L\tilde{C}\Vert x\Vert^L\right)^2}\\
            & =
            \operatorname{dist}_{\Vert\cdot\Vert}(x , \mathcal{C})\cdot \sqrt{m \left(2^L\tilde{C}\Vert x\Vert^L\right)^2}\\
            & =
            2^L\sqrt{m}\tilde{C}\cdot\Vert x\Vert^L \operatorname{dist}_{\Vert\cdot\Vert}(x , \mathcal{C})
        .\end{align*}
        The claim now follows with $U:=B_\delta (0)$ and $C:=2^L \sqrt{m} \tilde{C}$.
        \end{proof}
        
\end{appendices}

\end{document}